\documentclass[journal]{IEEEtran}
\usepackage{cite}

\ifCLASSINFOpdf
  \usepackage[pdftex]{graphicx}
\else
\fi
\usepackage{amsmath}
\usepackage{amsfonts}
\usepackage{amssymb}
\usepackage{color}
\usepackage{algorithmic}
\usepackage{algorithm}

\usepackage{array}

\ifCLASSOPTIONcompsoc
  \usepackage[caption=false,font=normalsize,labelfont=sf,textfont=sf]{subfig}
\else
  \usepackage[caption=false,font=footnotesize]{subfig}
\fi
\usepackage{fixltx2e}

\usepackage{stfloats}
\usepackage{url}

\usepackage{amsmath, amsthm, amssymb}
\newtheorem{corollary}{Corollary}
\newtheorem{proposition}{Proposition}
\newtheorem{lemma}{Lemma}

\renewcommand{\qedsymbol}{$\blacksquare$}
\usepackage{booktabs}
\usepackage{varwidth}
\usepackage{comment}
\usepackage{xpatch}
\makeatletter
\xpatchcmd{\proof}{\@addpunct{.}}{\@addpunct{:}}{}{}

\makeatother
\begin{document}

%
\title{Phase Shift-Free Passive Beamforming for Reconfigurable Intelligent Surfaces}
%
%
%

\author{Aymen~Khaleel,~\IEEEmembership{Graduate Student Member,~IEEE}, and
        Ertugrul~Basar,~\IEEEmembership{Senior Member,~IEEE}
\thanks{This work was supported by the Scientific and	Technological Research Council of Turkey (TUBITAK) under Grant 120E401.
	
The authors are with the Communications Research and Innovation Laboratory (CoreLab),  Department of Electrical and Electronics Engineering, Ko\c{c} University, Sariyer 36050, Istanbul, Turkey. \mbox{Email: akhaleel18@ku.edu.tr, ebasar@ku.edu.tr}}
}

\maketitle

\begin{abstract}
 Reconfigurable intelligent surface (RIS)-assisted communications recently appeared as a game-changing technology for next-generation wireless communications due to its unprecedented ability to reform the propagation environment. One of the main aspects of using RISs is the exploitation of the so-called passive beamforming (PB), which is carried out by adjusting the reflection coefficients (mainly the phase shifts) of the individual RIS elements. However, practically, this individual phase shift adjustment is associated with many issues in hardware implementation, limiting the RIS achievable gain. In this paper, we propose a low-cost, phase shift-free and novel PB scheme by only optimizing the on/off states of the RIS elements while fixing their phase shifts. The proposed PB scheme is shown to achieve the same scaling law (quadratic growth with the RIS size) for the signal-to-noise ratio as in the classical phase shift-based PB scheme, yet, with far less sensitivity to spatial correlation and phase errors. We provide a unified mathematical analysis that characterizes the performance of the proposed PB scheme and obtain the outage probability for the considered RIS-assisted system. Based on the provided computer simulations, the proposed PB scheme is shown to have a clear superiority over the classical one under different performance metrics.
\end{abstract}

%
\IEEEpeerreviewmaketitle
%
%
%
%
\section{Introduction}
\IEEEPARstart{R}{econfigurable}  intelligent surface (RIS)-based communications has recently emerged as a promising technology for beyond fifth-generation (5G) wireless networks due to its unprecedented capabilities to reform the communication environment \cite{Transmission_conference}. RISs contain an array of passive and low-cost elements, where each element can reflect (and modify) or fully absorb the electromagnetic waves impinging its surface. Due to their unique properties, the RIS literature has considered the integration of RISs almost to all existing wireless communication systems, making this emerging technology a potential game-changing factor for next-generation wireless communications systems. For example, in \cite{Transmission_conference}, an RIS is used to achieve an ultra-reliable wireless communication system, while in \cite{EB2}, an RIS is used along with index modulation to boost the data rate. RISs can also be used to replace (reduce) some of the RF chains as in \cite{AK1} and \cite{{Expermintal-RIS}}, where the modulation process is moved from the transmitter to the RIS side. Furthermore, in \cite{AK2} and \cite{mahmoud}, the authors used,  respectively, the RIS and the simultaneous transmitting and reflecting intelligent surface (STAR-RIS), to assist non-orthogonal-multiple-access (NOMA) networks through novel partitioning schemes.
\vspace{-0.2cm}
\subsection{Related Works}
\subsubsection{Studies on phase shift-based PB schemes}
In RIS-assisted systems, the PB design is achieved by independently adjusting the reflection coefficient associated with each RIS element to align the reflected signals in a particular direction. This can be achieved, for example, by using a tunable impedance that can be continuously adjusted through mixed-signal integrated circuits (ICs) \cite{phs-cont1} or multilayer surface design \cite{phs-cont2}. In vast majority of RIS-assisted systems discussed in the literature, the adopted classical PB design fixes the reflection amplitude to unity (to achieve full reflection power) and adjusts only the phase shift associated with each RIS element. The main practical issues related to the classical phase shift-based PB design are discussed next from the hardware implementation perspective.

It is possible, yet, practically challenging to adjust the phase of each RIS element continuously as this requires a sophisticated and expensive hardware design and more control pins connected to each element \cite{Toward-smart}, which overrides the essential features of RISs to be a low-cost and easy-to-deploy solution. Therefore, a discrete phase shift design is considered at implementation, where each RIS element is connected to multiple positive-intrinsic-negative (PIN) diodes, and the combinations of the on/off states (through control signals) of these diodes produce the different required phase shifts \cite{phs-disc}. However, since RISs are envisioned to be deployed with a large number (could reach thousands) of elements \cite{Expermintal-RIS}, \cite{amp-RFocus}, having multiple control lines for each RIS element to achieve multiple phase shift levels is still a practical issue as it may cause configuration overhead. A common approach to solve this issue is to partition the RIS into multiple sub-surfaces, where a common phase shift can be applied to all the elements within the sub-surface to reduce the overall needed number of control lines \cite{phs-group1}, \cite{phs-group2}, at the expense of some degradation in the PB performance. Another issue is the phase shift sensitivity to the incident wave's angle, where it is reported in \cite{Expermintal-RIS} and \cite{phs-incid} that this problem is inherent in the RIS design and can break the channel reciprocity of wireless RIS-assisted systems. Furthermore, the classical phase shift model adopted in most of the literature assumes the independent adjustment of the amplitude and phase reflection coefficients. However, based on the real implementation in \cite{phs-amp-exp1} and the mathematical modeling in \cite{main_pract}, this assumption does not hold practically, and there are always amplitude-phase-dependent variations. In particular, the authors in \cite{main_pract} showed that applying the classical phase shift design and ignoring the amplitude dependency may lead to a $5$ dB performance loss compared to the ideal case with no amplitude-phase dependency. Likewise, in studies of \cite{phs-frq-wide,phs-frq-mimo,phs-frq-emwm}, it is shown that the amplitude is both phase and frequency-dependent, which necessitates a special phase shift design to capture the correlation between all these variables. In order to mitigate the impact of the reflection amplitude-phase correlation, the authors in \cite{opt_ris} proposed a heuristic phase-shift-adjustment algorithm to trade off the phase alignment against the amplitude reflection loss. In the same context, the mutual coupling between RIS elements \cite{phs-mutual1} is another critical issue that affects the phase shift adjustment accuracy and performance. In particular, it is shown in \cite{phs-mutual2} that the actual capacitance, and consequently the phase shift of each RIS element is $45\%$ determined by its intended capacitance and $55\%$ by the capacitance of the neighboring element.

In \cite{phs-err2-let} and \cite{spatial_ris}, it is shown that the classical PB strategy, where the channel phases are all aligned to a particular direction, is not always the best one when considering practical system settings such as spatial correlation and phase errors. In particular, the authors in \cite{spatial_ris} show that the classical PB under the spatial correlation between RIS elements leads to performance degradation for the information transfer, while it is preferable for energy harvesting. Considering the phase errors, the authors in \cite{phs-err2-let} showed that a blind PB scenario occurs when using classical PB, particularly under high (uniform) phase error levels.  
\vspace{0.2cm}
\subsubsection{Studies on reflection amplitude-based PB schemes}
In \cite{main_pract} and \cite{phs-frq-wide}, RIS phase shifts are optimized by considering the phase-amplitude and frequency-amplitude dependent variations, respectively, while in \cite{phs-frq-mimo} are both phase and frequency-dependent amplitude variations considered. In \cite{amp-RFocus}, the authors perform PB by only turning on/off each element using a majority voting algorithm based on the received signal strength indicator (RSSI) measurements. In \cite{near_pas}, the authors considered an RIS architecture where a fixed arbitrary phase shift is assigned to each element at the manufacturing stage, while the PB is carried out by jointly optimizing the on/off states of the individual elements. In \cite{exp_ris}, the authors design their PB vector by optimizing both the amplitude and phase reflection coefficients. Finally, in \cite{amp-on-off1} and \cite{amp-on-off2}, each element is allowed to have on/off states that are jointly optimized with the phase shifts to obtain the PB vector. 
\vspace{-0.2cm}
\subsection{Motivation and Contributions}
In light of the earlier discussions on hardware implementation and system performance issues associated with the classical phase shift-based PB, we propose a low-complexity phase shift-free novel PB scheme for RISs. In the proposed PB scheme, to achieve full reflection power \cite{phs-amp-exp1}, \cite{main_pract}, the RIS elements are assumed to be designed with a common and fixed phase shift ($\pi$). Thus, to perform PB, we utilize the channel state information (CSI) of the RIS elements to activate only the elements enhancing the constructive combining effect of the reflected signals. Note that the proposed PB scheme is different than that of\cite{near_pas}, which still suffers from arbitrary amplitude reflection loss due to the phase-amplitude dependent variations \cite{main_pract}. Furthermore, unlike our proposed PB scheme, \cite{amp-RFocus} depends on the changes in the RSSI measurements to turn on/off the individual elements; however, it is challenging to track these changes on the level of a single element. Moreover, even when turning on/off a group of elements, the changes in the RSSI measurements cannot be tracked accurately due to the incoherent interference of the signals reflected from these elements. In addition, unlike the statistical CSI-based phase shift design methods \cite{S-CSI}, our proposed PB scheme is not meant to reduce the channel overhead but instead, aims to overcome the practical issues related to the classical continuous phase shift-based PB. The main contributions of this paper can be listed as follows:

\begin{itemize}
	\item We propose a low-complexity, phase shift-free and novel PB scheme to avoid hardware implementation and system performance degradation issues associated with the classical phase shift-based PB \cite{Toward-smart,phs-group1,phs-group2, phs-incid, phs-amp-exp1,main_pract,phs-frq-wide,phs-frq-mimo,phs-frq-emwm,phs-mutual1,phs-mutual2,phs-err2-let,spatial_ris}.
	\item We provide a unified mathematical analysis that characterizes the performance of the proposed PB scheme. First, we derive the asymptotic activation probability associated with each RIS element. Second, the lower bound for the average number of activated RIS elements is obtained. Third, we obtain lower and upper bounds for the probability of activating a certain number of RIS elements at each transmission. Forth, we show that the proposed scheme provides the same scaling law for the signal-to-noise ratio (SNR) as in the classical phase shift-based PB, where the SNR grows quadratically with the RIS size. This is achieved with a complexity level that is linear in the RIS size.
	\item We characterize the performance of the proposed RIS-assisted system by obtaining its outage probability and a tight upper bound to the ergodic rate.
	\item Under spatial correlation and phase errors, we provide comprehensive computer simulations to reveal the superiority of the proposed scheme over the two considered benchmark schemes \cite{main_pract} and \cite{opt_ris} in terms of the outage probability and ergodic rate performance.
\end{itemize}

 The rest of the paper is organized as follows. In Section II, we describe the system model and explain the concept of the proposed PB scheme. In Section III, we provide the mathematical analysis and performance evaluation of the proposed PB scheme. Computer simulations are given in Section V and the paper is concluded in Section VI.
\section{phase shift-Free PB: System Model}\label{sec:Main}
 Consider a downlink single-input single-output (SISO) system with an $N$-element reconfigurable intelligent surface (RIS) deployed close to the source (S), providing an alternative communication link to the blocked S-Destination (D) one, as shown in Fig. \ref{fig:system_block}(a). Let $\eta_n$ and $\theta_n$ denote the reflection amplitude and phase of the $n^{th}$ element, respectively, thus; the received signal at D is given as
 \begin{align}
 	y&=\sqrt{PL}\smash[b]{\sum}_{n=1}^{N}h_n\eta_n e^{j\theta_n} g_n x+v,\label{eq:y_received}
 \end{align}
 where $P$ and $x$ are the transmitted signal and power, respectively,  $L=L^{(S)}L^{(D)}$, $L^{(S)}$ and $L^{(D)}$ being the S-RIS and RIS-D path gains. $h_n=\alpha_ne^{-j\phi_n}$ and  $g_n=\beta_n e^{-j\psi_n}$ are the S-RIS ($n^{th}$ element) and RIS ($n^{th}$ element)-D channel coefficients, where ($\alpha_n$, $\phi_n$) and ($\beta_n$, $\psi_n$) are the phase and amplitude pairs of $h_n$ and $g_n$, respectively. $h_n$ and $g_n$  are assumed to be mutually independent and identically distributed (i.i.d.) random variables (RVs), in particular, $h_n, g_n\sim\mathcal{CN}(0,1)$, where $\mathcal{CN}(0,\sigma^2)$ stands for complex Gaussian distribution with zero mean and $\sigma^2$ variance. $v$ is the additive white Gaussian noise (AWGN) sample at D, $v\sim\mathcal{CN}(0,\sigma^2)$. By considering a practical RIS design (under the narrowband signal assumption), for the $n^{th}$ RIS element, the amplitude $\eta_n$ and phase $\theta_n$ reflection coefficients are correlated such that\footnote{Note that the circuit model in \cite{main_pract} is verified with the experimental results reported in \cite{phs-amp-exp1} and \cite{phs-amp-exp2}. Furthermore, \eqref{eq:theta_amp} is also adopted in \cite{opt_ris}, and it is applicable to many of the semiconductor devices usually used to build RISs \cite{main_pract}. Finally, an amplitude-phase-frequency correlation formulas similar to \eqref{eq:theta_amp} are provided in \cite{phs-frq-mimo}.}\cite{main_pract}
 \begin{figure}[t]
 	\subfloat[]{\includegraphics[width=45mm]{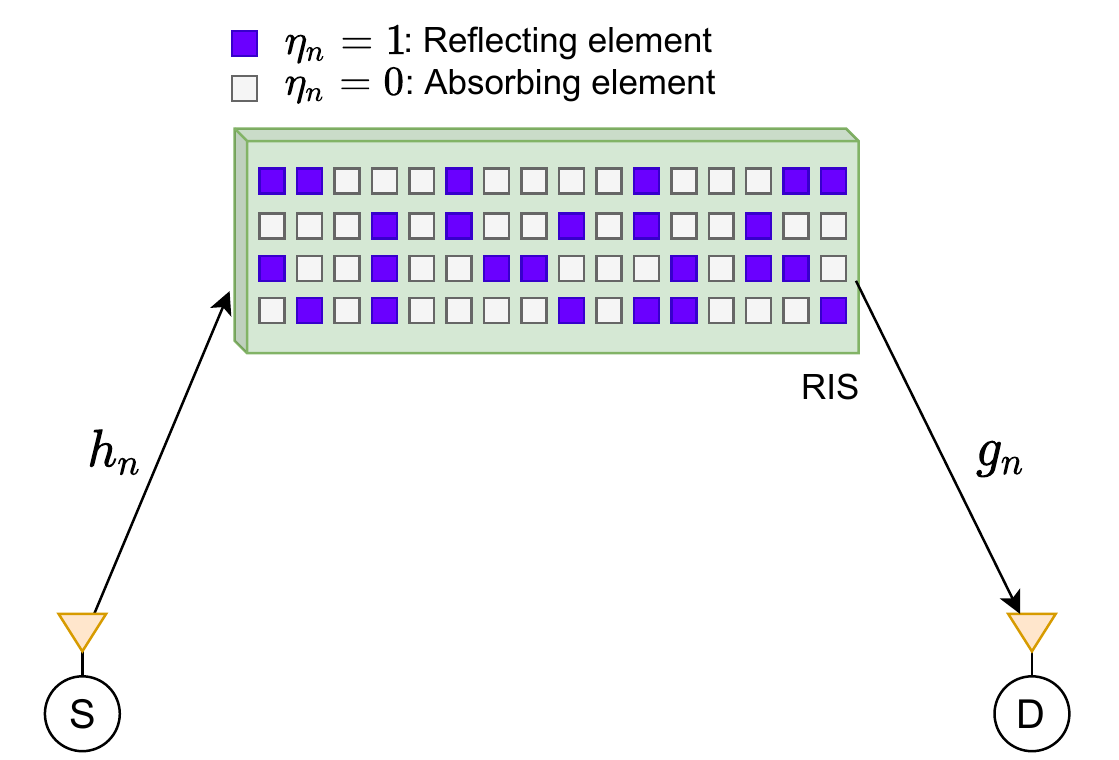}}	\subfloat[]{\includegraphics[width=35mm]{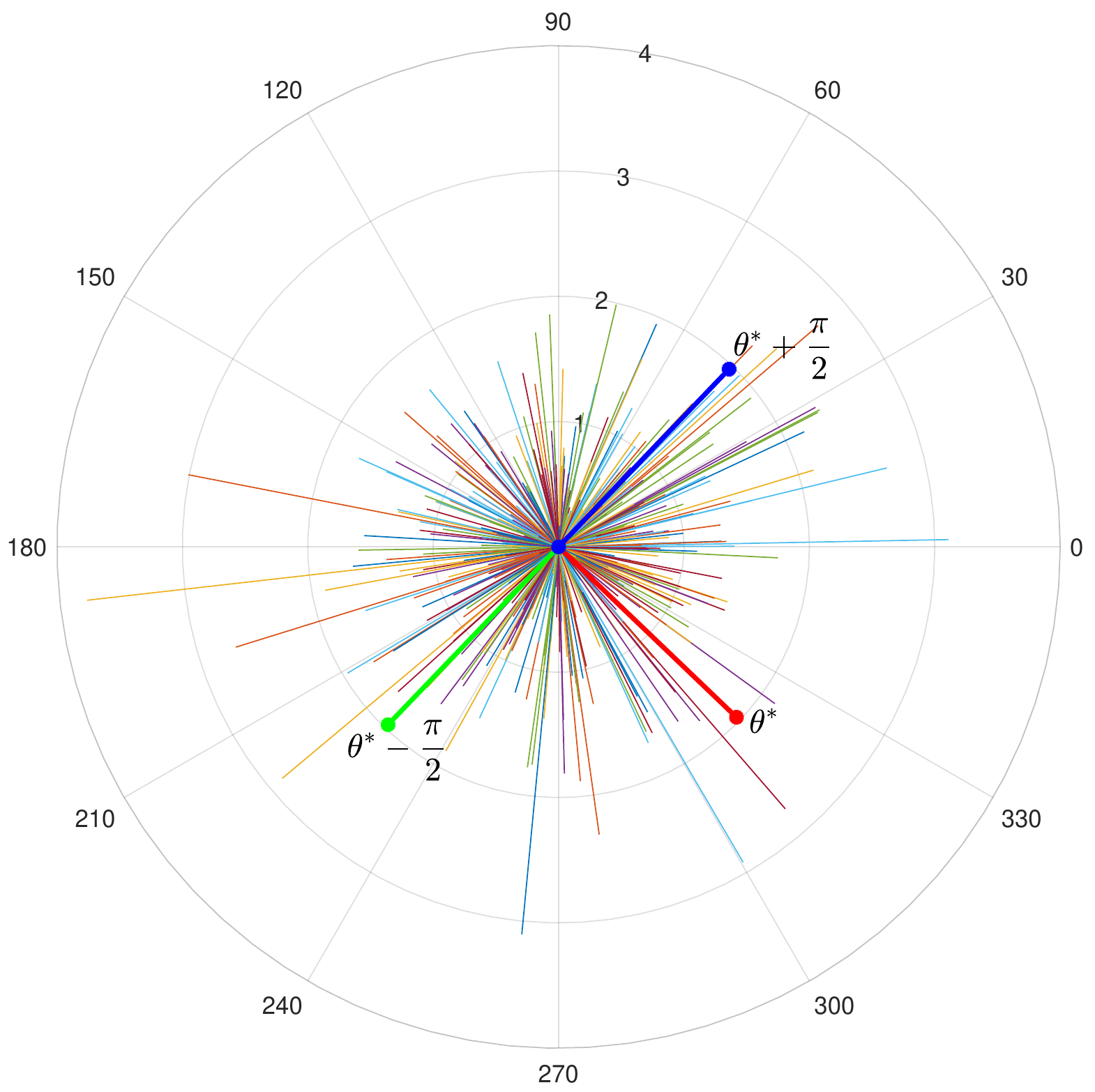}}
 	\caption{(a) RIS-assisted SISO communications system, (b) channel coefficients as 2D vectors (amplitudes and phases) in polar coordinates.}
 	\label{fig:system_block}
 \end{figure}
 \begin{align}
 	\eta_n(\theta_n)&=(1-a_{\min})\left(\frac{\sin(\theta_n-b_{hrz})+1}{2}\right)^{c_{stp}}+a_{\min},\label{eq:theta_amp}
 \end{align}
 where $a_{\min}\geq 0$, $b_{hrz}\geq 0$, and $c_{stp}\geq 0$ are the constants related to the specific circuit implementation.  $a_{\min}$ is the minimum amplitude, $b_{hrz}$ is the horizontal distance
 between $-\pi/2$ and $a_{\min}$, and $c_{stp}$ controls the steepness of the function curve \footnote{Note that when $a_{\min} = 1$ (or $c_{stp} = 0$), we obtain the ideal case of the phase shift model with full reflection and independent phase and amplitude adjustment. Practically, these parameters are determined at the RIS manufacturing stage, and then, these parameters can be obtained by a standard curve fitting tool.} \cite{main_pract}. In particular, we consider the RIS design in \cite{phs-amp-exp1}, which is also considered in \cite{opt_ris}, where we have $a_{\min}=0.2, b_{hrz}=0.43\pi$, and $c_{stp}=1.6$. 
 
 In the classical PB design, in order to maximize the SNR at D, the RIS phase shifts are adjusted such that $\theta_n=-(\phi_n+\psi_n)+e_n$, where $e_n$ is the residual phase error, and it is assumed to have mutually i.i.d. instances. Note that $e_n$ exists due to the phase estimation errors and/or quantization errors when producing the set of discrete phase shifts \cite{phs-err1-uni,phs-err2-let,phs-err3-let}. In particular, $e_n$ is assumed to have Von Mises distribution with a zero mean and a concentration parameter $\kappa$ \cite{von_mises}. On the other side, in the proposed scheme, we avoid all of the aforementioned problems by applying a fixed phase shift to all elements at the manufacturing stage of the RIS. Specifically, we apply a phase shift of $\pi$ for all elements ($\theta_n=\pi, \forall n$); hence, we guarantee the full reflection of the incident waves \cite{phs-amp-exp1}, \cite{main_pract}. In this way, in order to maximize the SNR, the on/off states of the individual RIS elements need to be jointly optimized to reach the best pattern that enhances the constructive combining of the signals received at D. It is not difficult to see that this optimization problem is of the non-convex combinatorial (binary) types which is a non-deterministic polynomial-time (NP) hard problem if we considered the exhaustive search solution \cite{NP}. Therefore, as follows, we propose a sub-optimal solution to find the on/off states pattern of the RIS elements that maximizes the SNR at D.
 
 In order to get an intuitive insight, we approach the problem from a geometrical perspective. More specifically, the S-RIS-D channel coefficients are complex numbers that can be represented as two-dimensional (2D) vectors in the complex plane, where the real and imaginary parts are the first and second dimensions, respectively. To obtain a geometrical insight, we represent the 2D vectors in polar coordinates as shown in Fig. \ref{fig:system_block}(b), where the length and angle of each vector correspond to the channel amplitude and the phase shift between the real and imaginary parts, respectively. Thus, finding the set of vectors that have the maximum length of the resultant vector associated with their sum is equivalent to finding the on/off state's pattern that maximizes the SNR at D. Therefore, we propose Algorithm 1 to find the aforementioned desired set of vectors, as follows. In Algorithm 1, at the first stage, we find the dominant direction ($\theta^*$) where most of the vectors are pointing towards it. Next, we consider only the vectors lie within the span of $\pm\frac{\pi}{2}$ from the dominant direction and activate ($\eta_n=1$) the RIS elements associated with them, see Fig. \ref{fig:system_block}(b). At the second stage, we seek any vector within the ones excluded at the first stage that can increase the overall sum obtained before. From \eqref{eq:y_received}, the instantaneous SNR is given by
\begin{align}
\text{SNR}&=\frac{\left|\sqrt{PL}\smash[b]{\sum}_{n=1}^{N}h_n\eta_n e^{j\theta_n} g_n\right|^2}{\sigma^2}\nonumber\\
&\stackrel{(a)}{=}\frac{\left|\sqrt{PL}\sum_{n^*=1}^{N_a}h_{n^*}g_{n^*}\right|^2}{\sigma^2},\nonumber\\
&=L\rho\bar{H}\;\label{eq:SNR}
 \end{align}
where $\rho=\frac{P}{\sigma^2}$ is the transmit SNR, $\bar{H}=|\sum_{n^*=1}^{N_a}h_{n^*}g_{n^*}|^2$. Furthermore, we obtain (a) by applying Algorithm 1, in particular, we set $\theta_n=\pi, \forall n$, where $n^*\in\mathcal{R}^*$ denotes the index of the activated RIS element and $N_a=|\mathcal{R}^*|$ is the total number of the activated elements. Here, $|\mathcal{R}^*|$ denotes the set cardinality. In what follows, we provide the steps of Algorithm 1 and give useful insights into its performance.
 \begin{algorithm}
 	\label{detctorA}
 	\caption{Phase shift-free PB}
 	\begin{algorithmic}[1]
 		\REQUIRE $h_n$, $g_n, \forall n$.
 		\STATE \textbf{Initialization:} $\mathbf{\Gamma}=\text{diag}(\eta_1, ..., \eta_N), \eta_n=0, \forall n$.
 		\STATE Obtaining the dominant direction:
 	 $\theta^*=\arg(\sum_{n=1}^{N}h_ng_n)$.
 	 	\STATE Specifying the endpoints of the activation region: $c_1=\theta^*-\frac{\pi}{2}, c_2=\theta^*+\frac{\pi}{2}$.
 		\STATE Let $\mathcal{R}^*=\{\emptyset\}, Z_n=\arg(h_ng_n)$, and $\Pi(s)$ be the function that maps the value of $s$ to the natural range of an angle, $[0,2\pi)$, as defined in \eqref{eq:pi_fun}.
 		\FOR{$n=1:N$}
 		\IF{$\Pi(c_1)\leq \Pi(Z_n)\leq \Pi(c_2)$}
 		\STATE $\eta_n=1$.
 		\STATE $\mathcal{R}^*=\mathcal{R}^*\cup\{n\}$.
 		\ENDIF
 		\ENDFOR
 		\STATE Update the entries of $\mathbf{\Gamma}$ according to the above loop, and let $\mathbf{h}=[h_1, ..., h_n, ..., h_N]$ and  $\mathbf{g}=[g_1, ..., g_n, ..., g_N]^T$, where $(\cdot)^T$ is the transpose operator.
 		\STATE Construct the new set $\tilde{{R}}=\{1, ..., N\}\setminus\mathcal{R}^*$.
 		\FOR {$m=1:|\tilde{{R}}|  $}
 		\STATE $n=\tilde{{R}}^{(m)}$.
 		\IF{$|\mathbf{h}\mathbf{\Gamma}\mathbf{g}+h_n\eta_ng_n|>|\mathbf{h}\mathbf{\Gamma}\mathbf{g}|$}
 		\STATE $\eta_n=1$.
 		\STATE $\mathcal{R}^*=\mathcal{R}^*\cup\{n\}$.
 		\ENDIF
 		\ENDFOR	
 		\RETURN $\eta_1, ..., \eta_N$.
 	\end{algorithmic}
\end{algorithm}
\section{Mathematical Analysis and Performance Evaluation of phase shift-Free PB}\label{sec:OP}
In this section, we characterize the performance of the proposed PB scheme by deriving the probabilities associated with the activation of each RIS element, the minimum number of activated elements at each transmission, the outage probability, and ergodic rate. Finally, we derive the scaling law associated with the growth of the SNR in proportion to $N$. Note that, in the following analyses, we assume no phase errors ($e_n=0$), and we consider it later in our computer simulations.
\begin{lemma}
The direction of the 2D vector associated with the $n^{th}$ S-RIS-D channel coefficient can be characterized by the random angle $Z_n=\arg(h_ng_n)$, which follows the triangle distribution with the following probability density function (PDF):
\begin{align}
f_{Z_n}(z)=
\begin{cases}
\frac{z+2\pi}{4\pi^2},&\;\;\; -2\pi\leq z< 0\nonumber \\
\frac{-z+2\pi}{4\pi^2},&\;\;\; 0\leq z< 2\pi\nonumber\\
0,&\;\;\;\text{otherwise}.\label{eq:z_pdf}
\end{cases}\\
\end{align}
\end{lemma}
\begin{proof}
Note that $\phi_n$ and $\psi_n$ are associated with the circular complex Gaussian (CCG) RVs $h_n$ and $g_n$, respectively, therefore we obtain $\phi_n,\psi_n\sim\mathcal{U}(-\pi,\pi)$ with $f_{\phi_n}(z)=f_{\psi_n}(z)=\frac{1}{2\pi}$, where $\mathcal{U}(a,b)$ denotes the uniform distribution over the interval $[a,b]$. Due to the independence of $\phi_n$ and $\psi_n$, the PDF of $Z_n=\arg(h_ng_n)=\phi_n+\psi_n$ can be found as follows:
\begin{align}
f_{Z_n}(z)=(f_{\phi_n}*f_{\psi_n})(z)=\int_{-\infty}^{\infty}f_{\phi_n}(\tau)f_{\psi_n}(z-\tau)d_{\tau}.\label{eq:conv}
\end{align}
Note that \eqref{eq:conv} corresponds to the convolution of a rectangular pulse with itself, which results in a triangular pulse over the interval $[-2\pi,2\pi]$ as in \eqref{eq:z_pdf}.  
  	\vspace{-0.25cm}
\end{proof}
Lemma 1 shows that the direction (phase) associated with the S-RIS-D  channel is concentrated around phase zero, for any given RIS element. In general, the S-RIS-D channel's phase is directly determined by the distribution of S-RIS and RIS-D channels. 
\begin{lemma}
	The dominant direction can be characterized as a random angle $\theta^*$ which, asymptotically, converges to the uniform distribution, that is, $\theta\xrightarrow{d}\mathcal{U}(-\pi,\pi)\;\text{as}\; N\rightarrow\infty$.
\end{lemma}
\begin{proof}
	The dominant direction can be found by obtaining the resultant vector of the sum of all the channels' coefficients (2D vectors) and then obtaining the angle associated with it. In light of Lemma 1, we obtain
	\begin{align}
		\theta^*=\arg\left(\sum_{n=1}^{N}w_ne^{jZ_n}\right),\label{eq:theta}
	\end{align}
where, $w_n=|h_ng_n|$, and the product $h_ng_n$ is the overall S-RIS-D channel. It can be seen that $\theta^*$ in \eqref{eq:theta} corresponds to the angle associated with a sum of $N$ weighted 2D unit vectors. Furthermore, irrespective of the distribution of $w_n$ and $Z_n$, according to the central limit theorem (CLT), the sum in \eqref{eq:theta} converges to the CCG distribution, $\sum_{n=1}^{N}w_ne^{jZ_n}\sim\mathcal{CN}(0,N)$, as $N\rightarrow\infty$. Consequently, due to the circular complex property of the sum, its argument is uniformly distributed, $\theta^*\in\mathcal{U}(-\pi,\pi)$, which completes the proof.
\end{proof}
The result obtained in Lemma 2 shows that when a large RIS is considered, the dominant direction tends to be equally likely in any direction. However, with a small RIS size, the dominant direction directly depends on the wights (channels' amplitudes) associated with their corresponding phases.
\begin{lemma}
Asymptotically, $\theta^*, Z_1, ..., Z_n, ..., Z_N$ can be assumed to be pairwise (but not mutually) independent RVs.
\end{lemma}
\begin{proof}
See Appendix A.
\end{proof}
Lemma 3 shows that as the number of RIS elements increases, the dominant direction is determined by the overall contribution of all the individual directions and not by a given single direction ($Z_n$).
 
In what follows, we use the results obtained in Lemmas 1-3 to derive the activation probability in Proposition 1. Note that, activating and deactivating the $n^{th}$ RIS elements equivalent to set its amplitude reflection coefficient as $\eta_n=1$ and $\eta_n=0$, respectively. Furthermore, the activation probability in Proposition 1 is derived for the first stage of Algorithm 1 (Steps 5-10), where it is challenging to derive it for the second stage (Steps 13-19) due to the complex distribution of $\mathbf{h}\mathbf{\Gamma}\mathbf{g}$, as will be discussed later in the proof of Proposition 3. Therefore, the obtained activation probability serves as a lower bound. Also, note that the following results on the RIS elements activation process might be important for the engineering design. In particular, the knowledge of the statistics of the activated/deactivated RIS elements gives useful insight on the electromagnetic wave interference (EMI) [39] associated with the RIS as a whole and the power dissipation/harvesting in autonomous RISs [40].
\begin{proposition}\label{prop:Proposition 1}
The probability of activating the $n^{th}$ RIS element (according to Algorithm 1) can be lower-bounded as
\begin{align}
P_a&=P(\Pi(Z_n)\in R)\geq\frac{1}{2},\;\forall n,\label{eq:p_act}
\end{align}
where the lower bound is achieved as $N\rightarrow \infty$. Here, $R$ is the activation (valid) region and $\Pi(s)$ is the function given by
\begin{align}
\Pi(s)=
\begin{cases}
s,\;\;\;0\leq s\leq 2\pi\nonumber\\
s+2\pi,\;\;\;s< 0\\
s-2\pi,\;\;\;s> 2\pi.\\
\end{cases}\label{eq:pi_fun}\\
\end{align}
Here, $\Pi(s)$ maps the value of $s$ to the natural range of an angle, $[0,2\pi)$.
\end{proposition}
\begin{proof}
	See Appendix B.
\end{proof}
Asymptotically, Proposition 1 shows that the probability of activating or deactivating any element is equally likely when the RIS  has a large size. This result coincides with Lemma 2, where, when the RIS size is large, the activation region (semi-circle) has an equally likely probability of being in any direction. Furthermore, for a finite $N$, the activation probability is always higher than $0.5$. Finally, note that the obtained lower bound of $P_a$ is still valid under the assumption of spatial correlation between RIS elements, which can be explained as follows. If the RIS elements are spatially correlated, which is practically true, the S-RIS-D channel coefficients (2D vectors) are more concentrated in a specific direction (angle), which increases the number of vectors within the activation window, $N_a$, and thus, we obtain a higher activation probability $P_a$.
\begin{corollary}
The average number of activated RIS elements can be lower-bounded as
\begin{align}
\text{E}[N_a]\geq\frac{N}{2},\label{eq:N_a}
\end{align} 
where $\text{E}[\cdot]$ is the expectation operator and the lower bound is achieved as $N\rightarrow \infty$.
\end{corollary}
\begin{proof}
The proof follows directly from Proposition 1, where $\text{E}[N_a]=\text{E}\left[\sum_{n=1}^N\eta_n\right]=\sum_{n=1}^NP_n=NP_a\geq\frac{N}{2}$.
\end{proof}
The result provided in Corollary 1 can be explained as follows. When the number of RIS elements is relatively small, the 2D vectors associated with these elements are more likely to be concentrated within a specific beam, which increases the activation probability above $0.5$ and consequently, $\text{E}[N_a]>\frac{N}{2}$. However, as the RIS size increases, the semi-circle activation region can be in any direction (Lemma 2) and the activation probability converges to $0.5$ (Proposition 1), therefore, $\text{E}[N_a]\rightarrow \frac{N}{2}$.
\begin{corollary}
 The RVs $\eta_1, ..., \eta_n, ..., \eta_N$ are pairwise independent, but not $m$-wise independent for $2< m\leq N$.
\end{corollary}
\begin{proof}
We show that $\eta_1, ..., \eta_n, ..., \eta_N$ are pairwise independent RVs by showing that their probability mass functions (PMFs), for any given pair of elements $n$ and $\tilde{n}$, are independent, $P(\eta_{\tilde{n}}=1|\eta_n=1)=P(\eta_{\tilde{n}}=1)$, where $n, \tilde{n}\in\{1, ..., N\}$, and $n\neq\tilde{n}$. Note that the activation probability of the $n^{th}$ element is the PMF of $\eta_n$.

Consider the activation region $R$ specified by the endpoints $[\Pi(c_1),\Pi(c_2)]$, where, without loss of generality, we have $\Pi(c_1)\leq\Pi(c_2)$, as described in the proof of Proposition 1. Therefore, the activation of the $n^{th}$ element ($\eta_n=1$) corresponds to $\Pi(c_1)\leq \Pi(Z_n)\leq\Pi(c_2)$, however; as it can be seen from the proof of Proposition 1 (Fig. \ref{Fig:Angles}), giving this information does not change the activation region $R$ for the $\tilde{n}^{th}$ element as $\Pi(c_1)$ and $\Pi(c_2)$ are separated by $\pi/2$ and, therefore, they still span the entire circle. This shows that the activation probabilities (PMF of $\eta_n$) for any pair of elements $n$ and $\tilde{n}$ are independent and we have that $\eta_n$ and $\eta_{\tilde{n}}$ are independent RVs, $\forall n, \tilde{n}\in\{1, ..., N\}$. However,
we show that these RVs are not $m$-wise independent for $2<m\leq N$, as follows. Consider, without loss of generality, the non-zero probability event where two elements are activated, for example $\eta_1=\eta_2=1$, with $\Pi(Z_1)\leq\Pi(Z_2)$. Consequently, it can be readily seen that given $\eta_1=\eta_2=1$ changes the activation region $R$ for the $\tilde{n}^{th}$ element as the endpoints $\Pi(c_1)$ and $\Pi(c_2)$ do not span the entire circle anymore; more specifically, the endpoints cannot span the interval $[\Pi(Z_1),\Pi(Z_2)]$. This shows that $P(\eta_{\tilde{n}}=1|\eta_1=\eta_2=1)\neq P(\eta_{\tilde{n}}=1)$, thus,  $\eta_1,\eta_2$, and $\eta_{\tilde{n}}$ are non-independent RVs. Consequently, the RVs $\eta_1,..., \eta_N$ are not triple-wise independent, thus, not $m$-wise independent for $2<m\leq N$.
\vspace{-0.2cm}
\end{proof}
Corollary 2 shows that the activation probabilities of any two given elements are independent. However, the activation probabilities for $N>2$ elements are dependent. The result given in Corollary 2 is necessary for the next proposition, where we propose a new physical resources-based outage probability (ROP) metric. The ROP metric is defined to be the probability that the number of activated RIS elements at each transmission is less than a predefined threshold ($N_{thr}$), $\text{ROP}=P(N_a\leq N_{thr})$.
\begin{proposition}
For a given $N_{thr}$, we have $\text{ROP}_L\leq\text{ROP}\leq\text{ROP}_U$, where the lower and upper bounds can be given, respectively, as
\begin{align}
\text{ROP}_L&=1-U,\label{eq:ROP_U}\\
\text{where,}\nonumber\\
U&=
\begin{cases}
	1,\;\;\;N_{thr}<\bar{N}P_a\\
(\bar{N}\bar{P_a}+N_{thr})P_a/N_{thr},\;\bar{N}P_a\leq N_{thr}\leq 1+\bar{N}P_a\\
	\frac{N\bar{N}P_a^2+(\zeta-1)(\zeta-2NP_a)}{(N_{thr}-\zeta)^2+(N_{thr}-\zeta)},\;\;\;N_{thr}\geq 1+\bar{N}P_a\\
	 \zeta=\lceil\frac{NP_a(N_{thr}-1-\bar{N}P_a)}{N_{thr}-NP_a}\rceil,
\end{cases}\\
	\text{ROP}_U&=\sum_{i=0}^{N_{thr}} \binom{N}{i}P_a^i(1-P_a)^{N-i},\label{eq:bion-cdf}
 \end{align}
where $\bar{N}=N-1$ , $\bar{P}_a=1-P_a$, $\lceil\cdot\rceil$ denotes the ceiling function, and $U$ corresponds to the upper bound probability of $N_a$ exceeding $N_{thr}$.
\end{proposition}
\begin{proof}
Note that, from Corollary 2, $N_a=\sum_{n=1}^N\eta_n$ corresponds to a sum of pairwise independent Bernoulli RVs. Therefore, for $\text{ROP}_L$, $U$ corresponds to the probability of the sum of pairwise independent Bernoulli RVs ($N_a$) to exceed an integer threshold ($N_{thr}$), where the proof for $U$ is given in \cite[Theorem 4.1]{N_a_upper}.  

The upper bound $\text{ROP}_U$ is found by approximating $N_a$ to the Binomial distribution, as follows. Note that the sum of pairwise independent Bernoulli RVs can be approximated to the mutually independent case (Binomial distribution) with a total variation distance of $1-1/e$, as $N\rightarrow \infty$ \cite{binom-aprox}. Furthermore, it is worth noting that the approximation to the mutual independence case corresponds to the minimum activation probability, where the dominant direction $\theta^*$ is not a function of the individual directions $Z_n$ anymore. In light of these, the approximation of $N_a$ to the Binomial distribution serves as a lower bound for $N_a$ and, therefore, its CDF given in \eqref{eq:bion-cdf} serves as the upper bound for ROP.  
\end{proof}
In order to get useful insight from Proposition 2, we use Hoeffding's inequality (in light of (\ref{eq:bion-cdf})), $P(N_a-\text{E}[N_a]\geq c)\leq e^{-2\frac{c^2}{N}}$, for $c\geq0$, which shows that the number of activated RIS elements is concentrated around its mean value ($N/2$) and thus, ROP decreases as $N_{thr}$ decreases below the mean value of $N_a$.
\begin{proposition}
For a targeted transmission rate $r$, the outage probability $P_{out}$ of $D$ is given by
\begin{align}
P_{out}=\frac{1}{2}\left[1+\text{erf}\left(\frac{\ln(\bar{r})-\mu}{\sqrt{2}\bar{\sigma}}\right)\right],\label{eq:log-normal-cdf}
\end{align}
where $\bar{r}=\frac{2^r-1}{L\rho}$, $\mu=a_1N^{b_1}+c_1$, $\bar{\sigma}=a_2N^{b_2}+c_2$, and $\text{erf}(\cdot)$ is the error function. Considering the spatial correlation between RIS elements with inter-element separation of $\frac{\lambda}{8}$, we have $(a_1,a_2)=(-533.1,2.928), (b_1,b_2)\hspace{-0.1cm}=\hspace{-0.1cm}(-0.003336,-0.1783)$, and $(c_1,c_2)=(532.3,-0.6076)$. On the other side, ignoring spatial correlation due to inter-element separation much larger than $\frac{\lambda}{2}$ \cite{spatial-corr}, we have $(a_1,a_2)=( 39.59,1.725), (b_1,b_2)=(0.03871,-0.3917)$, and $(c_1,c_2)=(-40.54,-0.0354)$.
\end{proposition}
\begin{proof}
Considering a targeted transmission rate $r$ then, the outage probability is given as
\begin{align}
	P_{out}&=P(\log_2(1+\text{SNR})<r ),
\end{align}
from \eqref{eq:SNR}, we obtain
\begin{align}
	P_{out}&=P\left(\bar{H}<\frac{2^r-1}{L\rho}\right)\;\nonumber\\
	&=F_{\bar{H}}\left(\frac{2^r-1}{L\rho}\right),\label{eq:OP}
\end{align}
where $F_{\bar{H}}$ is the CDF of $\bar{H}$.
\begin{figure}[t]
	\centering  
	\subfloat[]{\label{fig2:a}\includegraphics[height=38mm, width=44mm]{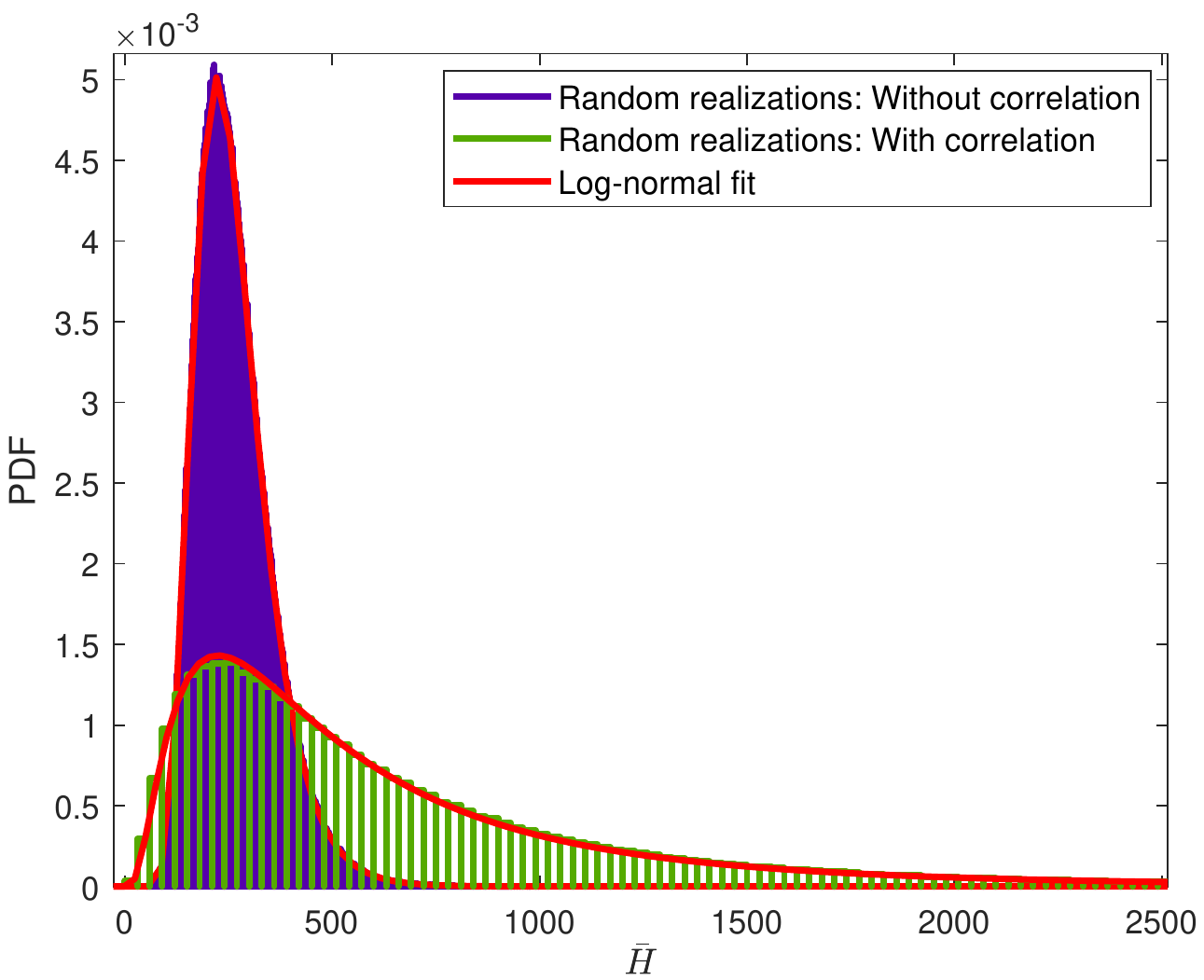}}
	\subfloat[]{\label{fig2:b}\includegraphics[height=38mm, width=45mm]{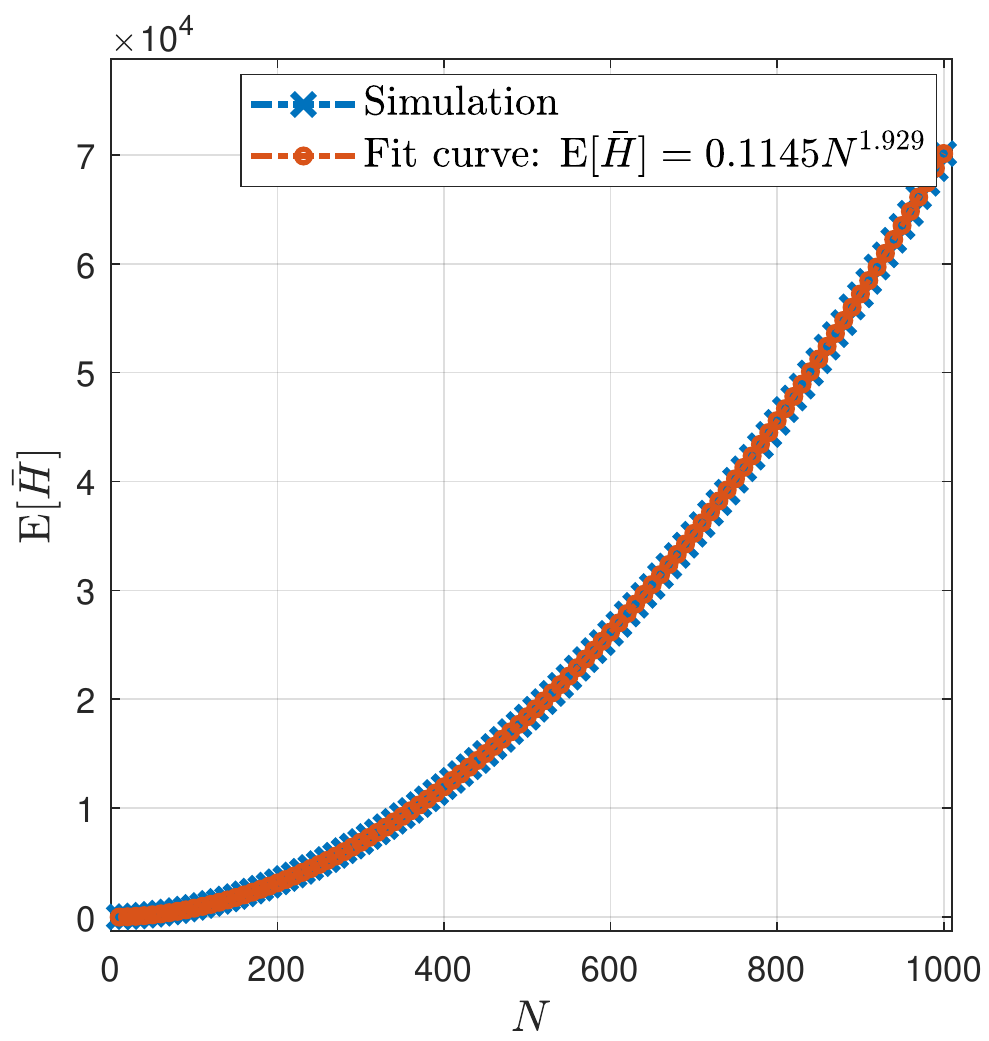}}
	\caption{(a) Fitting the distribution of $\bar{H}$ to the log-normal distribution and (b) the quadratic scaling law of the channel gain $\bar{H}$.}
	\label{fig:log-normal-fit}
	\vspace{-0.4cm}
\end{figure}
Note that the S-RIS-D channel coefficients $h_{n^*}g_{n^*}$, $n^*\in\{1, ..., N_a\}$ are correlated through their phases $\Pi(Z_{n^*})$, where $\smash{\underset{n^*}{\text{max}}}(\Pi(Z_{n^*}))\leq\text{max}(\Pi(c_1),\Pi(c_2))$ and $\smash{\underset{n^*}{\text{min}}}(\Pi(Z_{n^*}))\geq\text{min}(\Pi(c_1),\Pi(c_2))$, as demonstrated in the proof of Proposition 1. Accordingly, the terms of the sum of $\bar{H}$ are non-independent, where each term is a product of two Gaussian RVs. Consequently, it is challenging to analytically derive the distribution of $\bar{H}$. Therefore, we use the Distribution Fitting Tool in MATLAB\footnote{Note that this tool uses a maximum likelihood estimator to estimate the assumed distribution's parameters, and the quality of the distribution fit can be measured by the standard error for the estimated parameters and the shape match of the PDFs.} to fit the distribution of $\bar{H}$ and, accordingly, obtain the outage probability, which is the CDF of $\bar{H}$. Fig. \ref{fig:log-normal-fit}(a) shows that the distribution of $\bar{H}$, with/without spatial correlation, perfectly matches the one of the log-normal RV with $\mu$ and $\bar{\sigma}$ are the mean and standard deviation of $\ln(\bar{H})$, respectively, $\ln(\bar{H})\sim\mathcal{N}(\mu,\bar{\sigma}^2)$. Consequently, we obtain the outage probability as the CDF of $\bar{H}$, as given in \eqref{eq:log-normal-cdf}. Furthermore, using the Curve Fitter tool in MATLAB, a semi-analytical representation is obtained for $\mu$ and $\bar{\sigma}$ as functions of $N\in\{10,500\}$, where the fitting parameters are given after (\ref{eq:log-normal-cdf}).
\end{proof}
Proposition 3 shows that the outage probability decreases as $\rho$ and/or $N$ increases. In particular, it can be readily verified that, for $N\in\{10,500\}$, we have $\mu, \bar{\sigma}>0$, regardless of the spatial correlation assumption. Thus, considering the asymptotic behavior with respect to $\rho$, as $\rho\rightarrow\infty$ ($\bar{r}\rightarrow 0$), we obtain $\text{erf}(\frac{\ln(\bar{r})-\mu}{\sqrt{2}\bar{\sigma}})\rightarrow -1$ and $P_{out}\rightarrow 0$. Likewise, considering the asymptotic behavior with respect to $N$, note that as $N\rightarrow 500$ we have $\mu\rightarrow 10$, $\bar{\sigma}\rightarrow 0.36$, accordingly, it can be readily shown that for $r<\log_2(L\rho e^{8.5}+1)$, we have $\text{erf}(\frac{\ln(\bar{r})-\mu}{\sqrt{2}\bar{\sigma}})\rightarrow -1$, and $P_{out}\rightarrow 0$.	
\begin{proposition}
The ergodic rate of $D$ can be upper-bounded as
\begin{align}
R\leq\log_2(1+L\rho\exp(\mu+\frac{\bar{\sigma}^2}{2})).
\end{align}
\end{proposition}
\begin{proof}
The ergodic rate of $D$ is given by
\begin{align}
R&=\text{E}[\log_2(1+\text{SNR})],\label{eq:rate}
\end{align}
by considering the concavity of the function in (\ref{eq:rate}), we use Jensen's inequality to obtain
\begin{align}
R&\leq\log_2(1+\text{E}[\text{SNR}]),\nonumber\\
&=\log_2(1+L\rho\exp(\mu+\frac{\bar{\sigma}^2}{2})),\label{eq:E_rate}
\end{align}
where, from (\ref{eq:SNR}), we have $\text{E}[\text{SNR}]=L\rho\text{E}[\bar{H}]$, furthermore, from the proof of Proposition 3, we have $\bar{H}$ has log-normal distribution with a mean $\text{E}[\bar{H}]=\exp(\mu+\frac{\bar{\sigma}^2}{2})$. Here, we consider the same semi-analytical representation for $\mu$ and $\bar{\sigma}$ given after (\ref{eq:log-normal-cdf}).
\end{proof}
Note that, from the semi-analytical representation for $\mu$ and $\bar{\sigma}$ given in the proof of Proposition 3, it can be readily verified that $\mu,\bar{\sigma}>0$ and the sum $\mu+\frac{\bar{\sigma}^2}{2}$ increases with $N$, regardless of the spatial correlation assumption. This shows that the ergodic rate is an increasing function in $N$. 
\begin{proposition}
An instantaneous SNR in order of $\mathcal{O}(N^2)$ can be achieved by Algorithm 1 with a complexity level of $\mathcal{O}(N)$.
\end{proposition}
\begin{proof}
	See Appendix C.
\end{proof}
\vspace{-0.2cm}
Proposition 5 shows that the proposed phase shift-free PB scheme preserves the quadratic growth behavior of the beamforming gain associated with the classical phase shift-based one, see Fig. 2(b). This can be explained as follows. Note that, although the signals reflected from the RIS (according to Algorithm 1) cannot be coherently combined, yet, the overall amplitude of their sum still linearly grows with the number of activated elements $N_a$. Also, by noting that $N_a$ grows with $N$, the overall S-RIS-D channel gain ($\bar{H}$) and, thus, the SNR grows quadratically with the RIS size $N$. Furthermore, contrary to the classical phase shift-based PB, the proposed PB scheme preserves the full power of the reflected signals at the expense of perfectly aligning them. This trade-off, between fully reflecting the signals and perfectly aligning them, gives the superiority to the proposed scheme over the classical one under phase error conditions, where the perfect alignment is inherently impossible.
\begin{figure}[t]
	\centering  
	\subfloat[]{\label{fig2:a}\includegraphics[width=45mm]{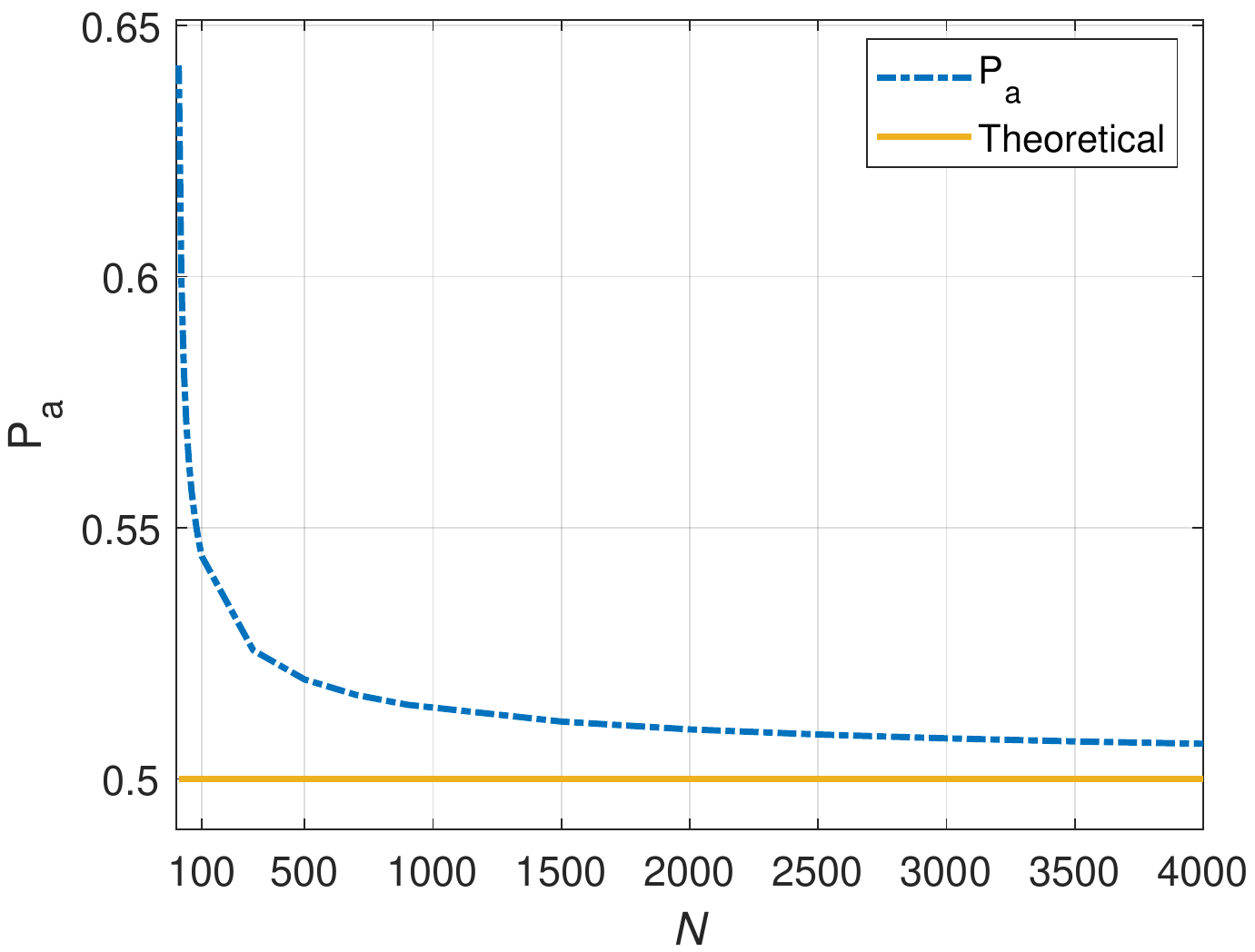}}
	\subfloat[]{\label{fig2:b}\includegraphics[width=45mm]{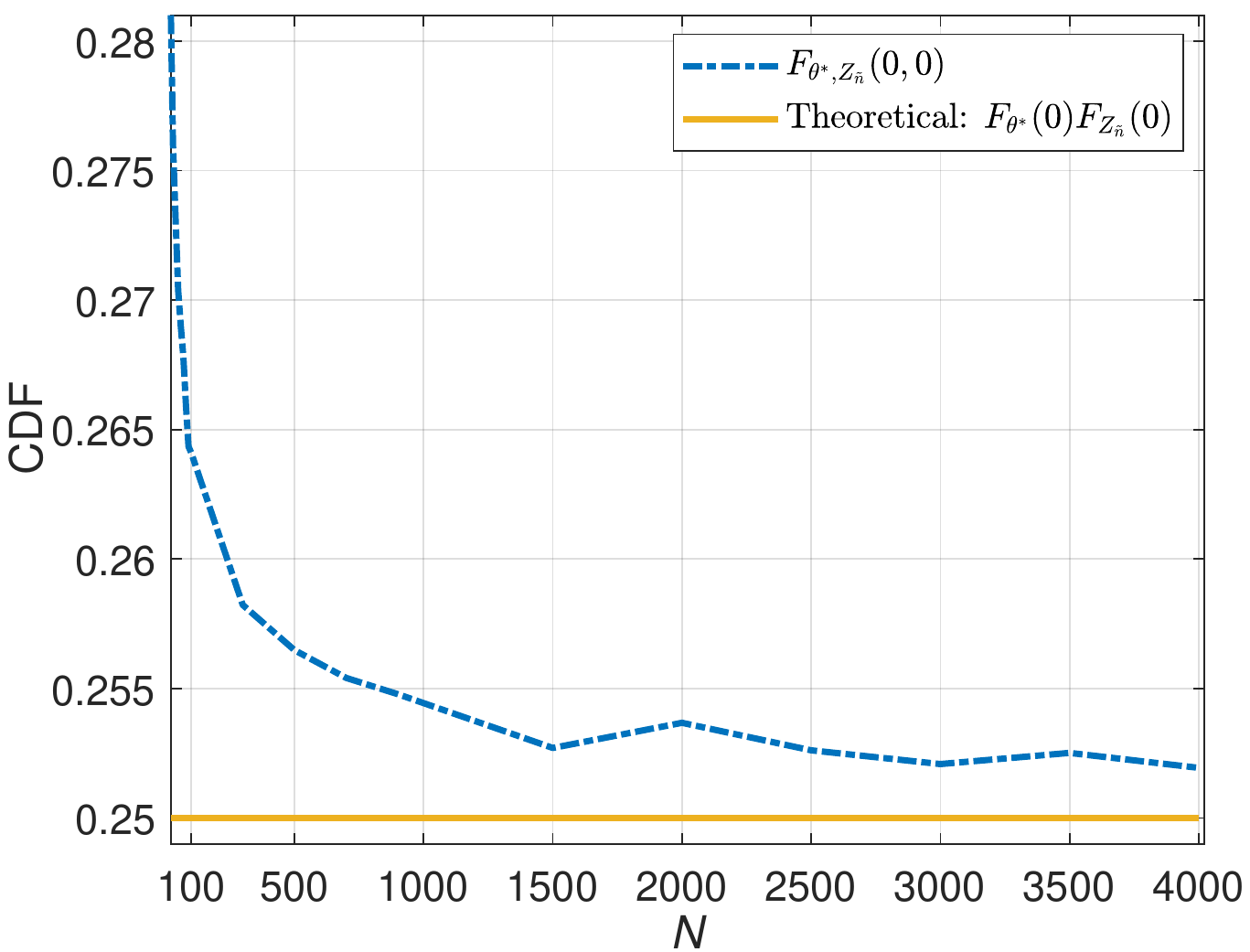}}
	\caption{(a) The activation probability and (b) the joint CDF $F_{\theta^*,Z_n}(0,0)$, both versus $N$.}
	\label{fig:P_a_cdf}
	\vspace{-0.4cm}
\end{figure}
\section{Simulation Results}\label{sec:Simu}
This section presents comprehensive computer simulations to examine the performance of both the proposed and benchmark schemes under different system settings. In particular, we consider the system performance under phase errors and/or spatial correlation. The adopted RIS size is $N=40$, where each element has a width $d_H$ and a length $d_V$, and we used the spatial correlation model presented in \cite{spatial-corr}. Note that, by default, no spatial correlation or phase errors are assumed in our simulations unless otherwise stated. We consider two different benchmark schemes; first, the classical PB (labeled with classical PB), where the applied phase shift on each RIS element is meant to remove the S-RIS-D channel phase associated with that element.
Second, the reflection phase selection algorithm (RPSA),  which is proposed in \cite{opt_ris} to mitigate the reflection phase and amplitude correlation issue, is considered. Note that, for the benchmark schemes, the reflection phase and amplitude are correlated through \eqref{eq:theta_amp} \cite{main_pract}. Also, for the benchmark schemes, two phase shift levels are adopted ($0$ and $\pi$) for a fair comparison with the proposed scheme that uses one control bit for the on/off adjustment of each RIS element. The RIS-D distance is $r_D=10$ m, and the S-RIS distance is given by $r_S=\lceil\frac{N\lambda}{2}\rceil$\cite{Expermintal-RIS} to ensure that the RIS is in the far-field region of S, where $\lambda$ is the wavelength associated with the operating frequency ($1.8$ GHz). The overall S-RIS-D path loss is given by $(L)^{-1}=\lambda^4/(256\pi^2r_S^2r_D^2)$ \cite{Ellingson}, and the noise level is chosen to be $\sigma^2=-90$ dBm.

Note that, considering the RPSA algorithm, it should be noted that both the RPSA and the proposed Algorithm 1 achieve the quadratic growth of the SNR with $N$ and requires the same complexity level, $\mathcal{O}(N)$. Furthermore, in all of the provided simulations, the proposed scheme activates, on average, around $60\%$ of the RIS elements, while, as stated by Proposition 1, the asymptotic number of activated  RIS elements is $N_a=N/2$. This effectively reduces the unavoidable EMI associated with each RIS element \cite{EMI-ris}. Furthermore, for autonomous RISs \cite{ris-auto}, the energy consumption associated with the control circuit of the RIS becomes critical. In this regard, our proposed scheme can save more than $50\%$ of the operating energy through the off-state elements. Also, using the energy harvesting strategy provided in \cite{ris-auto}, the off-state elements can be used as energy harvesters. In light of these, the proposed phase shift-free PB scheme makes the use of the RIS more energy-efficient with less reflection to the EMI interference compared to the benchmark schemes.

\begin{figure}[t!]
	\subfloat[]{\includegraphics[width=44mm,height=44mm]{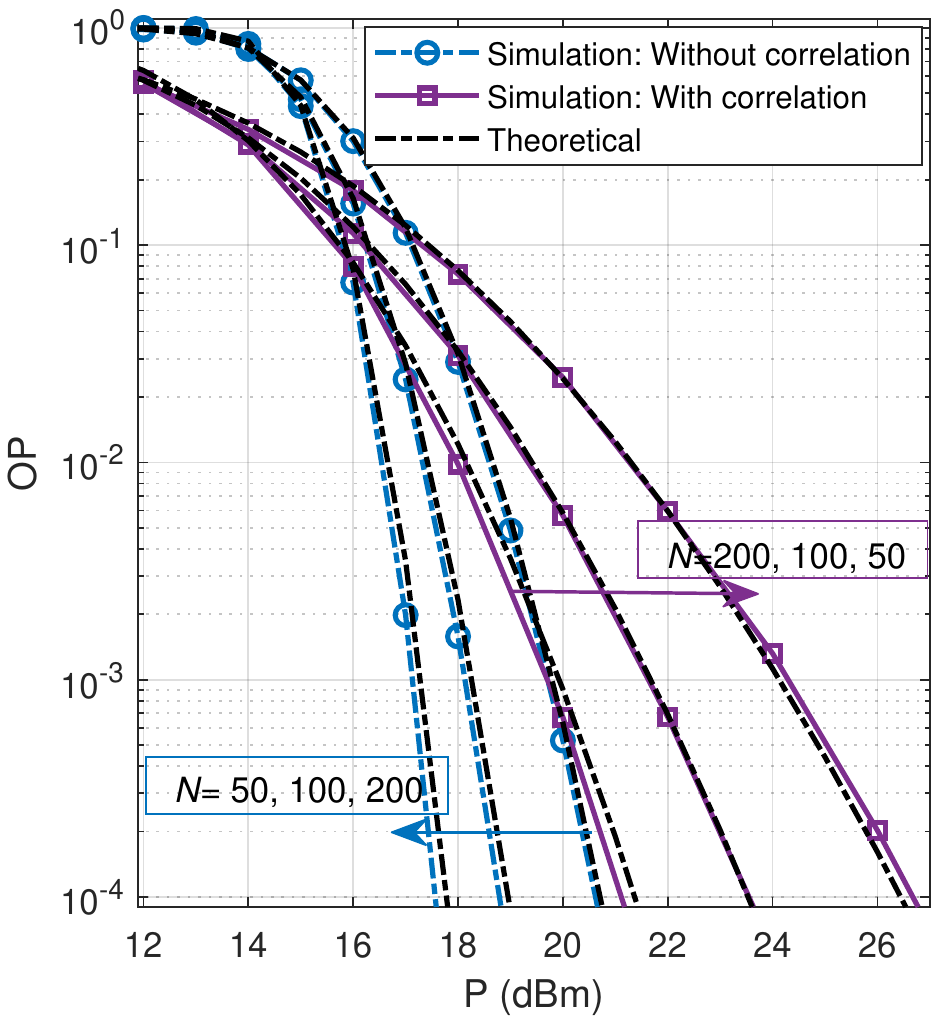}}	\subfloat[]{\includegraphics[width=44mm,height=44mm]{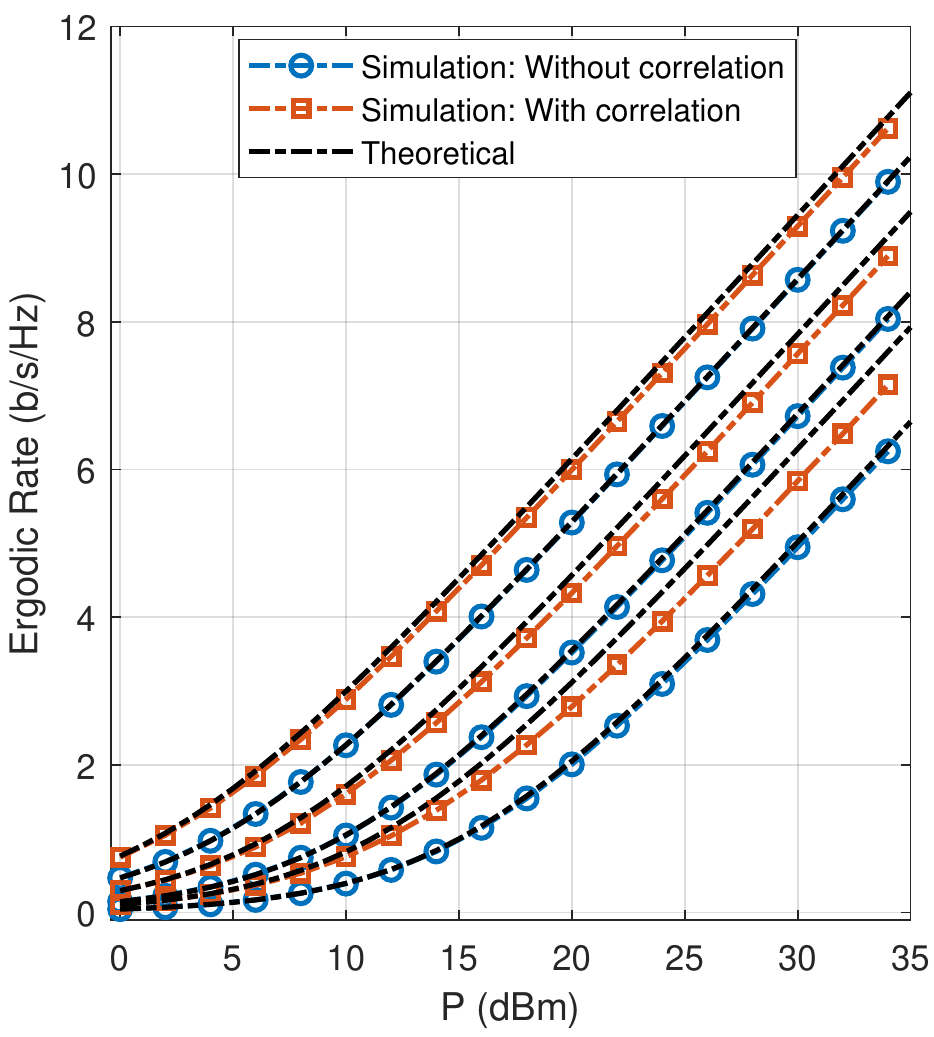}}
	\caption{The performance, with/without spatial correlation and with different $N$ values, of (a) outage probability with $r=1$ bit per channel use (bpcu), and (b) ergodic rate.}
	\label{fig:OP_R_theo}
	\vspace{-0.6cm}
\end{figure}

In Figs. \ref{fig:P_a_cdf}(a) and (b), we verify  Proposition 1 and Lemma 3, respectively. In particular, in Fig. \ref{fig:P_a_cdf}(a), the activation probability $P_a$ converges to the asymptotic value ($0.5$) as $N$ increases. In Fig. \ref{fig:P_a_cdf}(b), without loss of generality, we evaluate the joint CDF of $\theta^*$ and $Z_{\tilde{n}}$ at the point ($\theta=0,z=0$), where it can be seen that, asymptotically, the product of the individual CDFs converges to the joint CDF. 

Fig. \ref{fig:OP_R_theo}(a) shows the outage probability (OP) performance with different $N$ values, where increasing the RIS size improves the performance effectively. Furthermore, Fig. \ref{fig:OP_R_theo}(a) verifies Proposition 3, where the theoretical curves converge to those obtained through simulations as $N$ increases, with a perfect match at $N=200$. Furthermore, it can be seen that the OP performance deteriorates under spatial correlation due to the lack of diversity. Likewise, Fig. \ref{fig:OP_R_theo}(b) verifies Proposition 4, where the theoretical upper bound is shown to be tight and has a close match to the exact one. Furthermore, the ergodic rate is shown to be an increasing function in $N$, with, notably, higher performance under spatial correlation, which can be explained as follows. Under spatial correlation, the 2D vectors (S-RIS-D channel coefficients) are more concentrated in a specific direction (angle), which increases the number of vectors within the activation window, $N_a$. Consequently, as $N_a$ increases $\text{SNR}$ quadratically increases as stated in Proposition 5, which leads to a better performance in the ergodic rate. Note that, in Fig. \ref{fig:OP_R_theo}(b), in order to show the effect of increasing $N$ only, we ignore the distance effect by obtaining $r_s$ for $N=50$ and use it for the other values of $N$.

In Figs. \ref{fig:OP_err}(a) and (b), we show the outage probability performance with/without phase error and spatial correlation effects, as follows. In Fig. \ref{fig:OP_err}(a), it can be seen that the proposed scheme achieves a  $2$ dB gain in the required $P$ compared to the classical PB, while it is slightly better than RPSA. In Fig. \ref{fig:OP_err}(b), we show the combined effect of spatial correlation and phase errors, where the proposed scheme achieves around $5$ dB performance gain in the required $P$ compared to both benchmark schemes. 

In Figs. \ref{fig:Rate_Err}(a) and (b), we show the ergodic rate performance for all schemes with/without phase errors and spatial correlation effects, as follows. In Fig. \ref{fig:Rate_Err}(a), all schemes achieve almost the same performance, where neither phase errors or spatial correlation is considered. In Fig. \ref{fig:Rate_Err}(b), the combined effect of spatial correlation and phase errors is shown, where the proposed scheme achieves a remarkable performance gain of almost $5$ dB in the required $P$ compared to the benchmark schemes.
\begin{figure}[t!]
	\subfloat[]{\includegraphics[width=44mm,height=44mm]{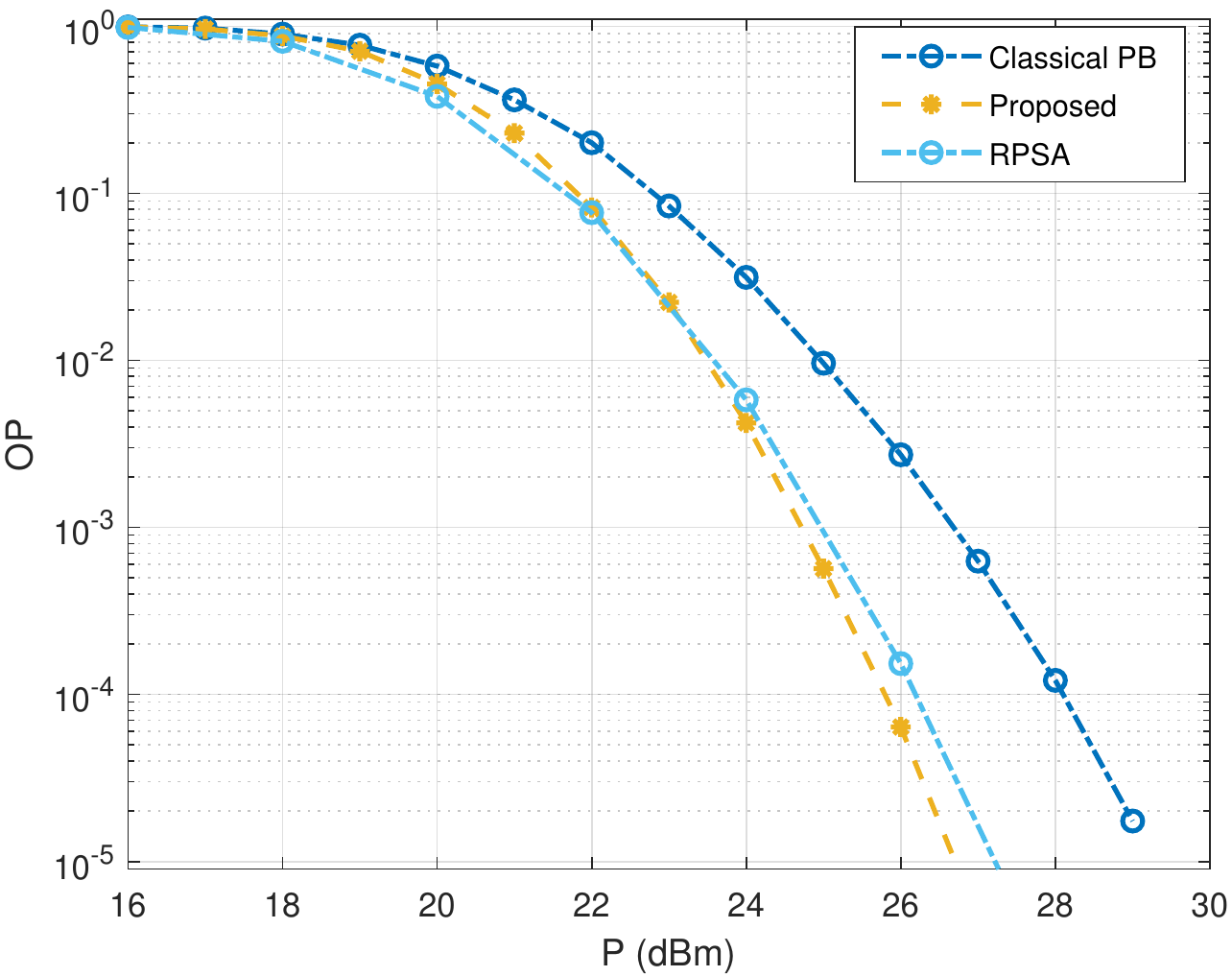}}	\subfloat[]{\includegraphics[width=44mm,height=44mm]{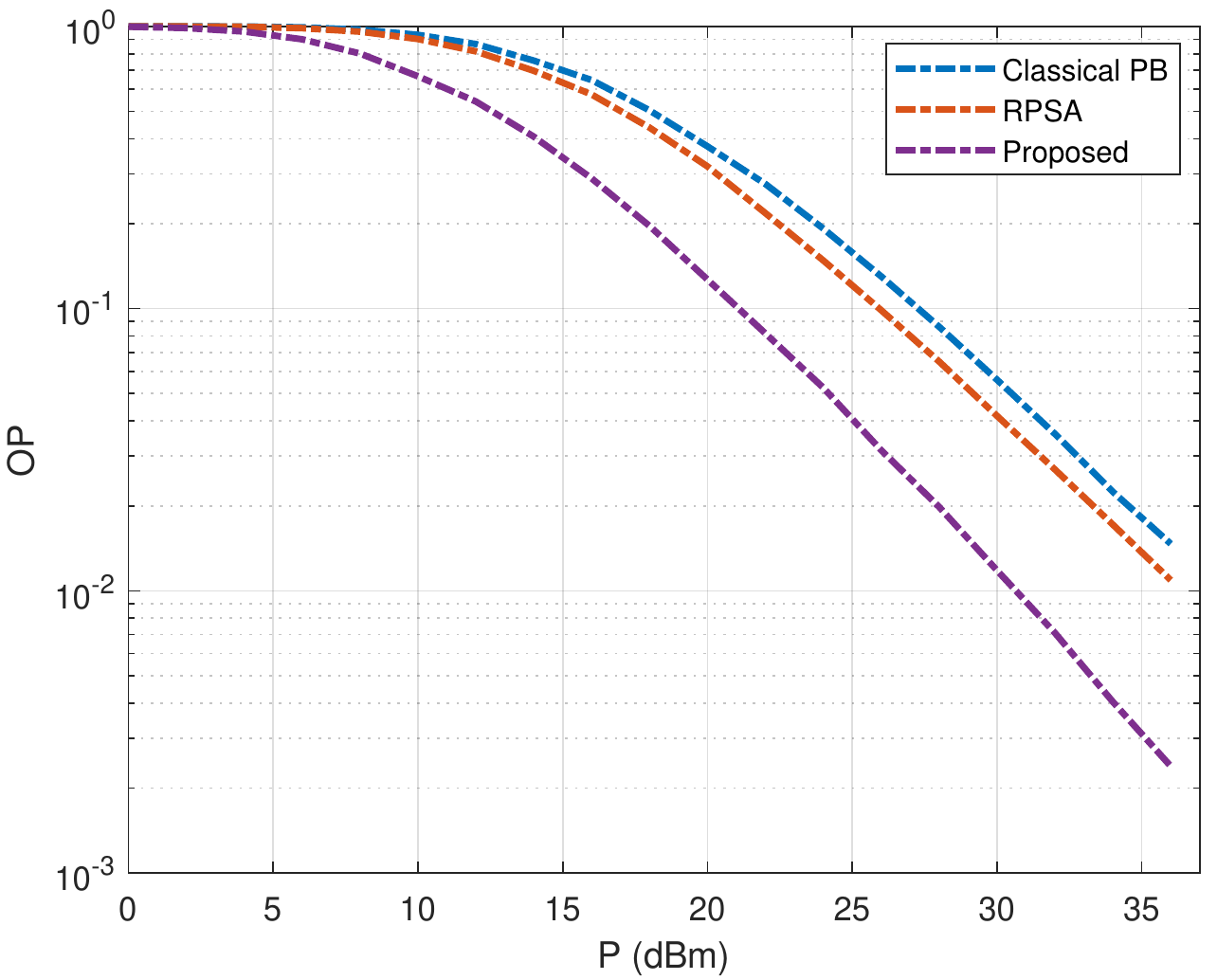}}
	\caption{Outage probability performance with (a) no phase errors, no spatial correlation, and $r=2$ bpcu, (b) spatial correlation ($d_H=d_V=d=\lambda/8$) \cite{amp-RFocus}, \cite{spacing_ele2}, \cite{spacing_ele1}, phase errors ($\kappa=0$) \cite{phs-err2-let}, \cite{unif-err}, and $r=0.5$ bpcu. }
	\label{fig:OP_err}
	\vspace{-0.6cm}
\end{figure}
\begin{figure}[t!]
	\subfloat[]{\includegraphics[width=44mm,height=44mm]{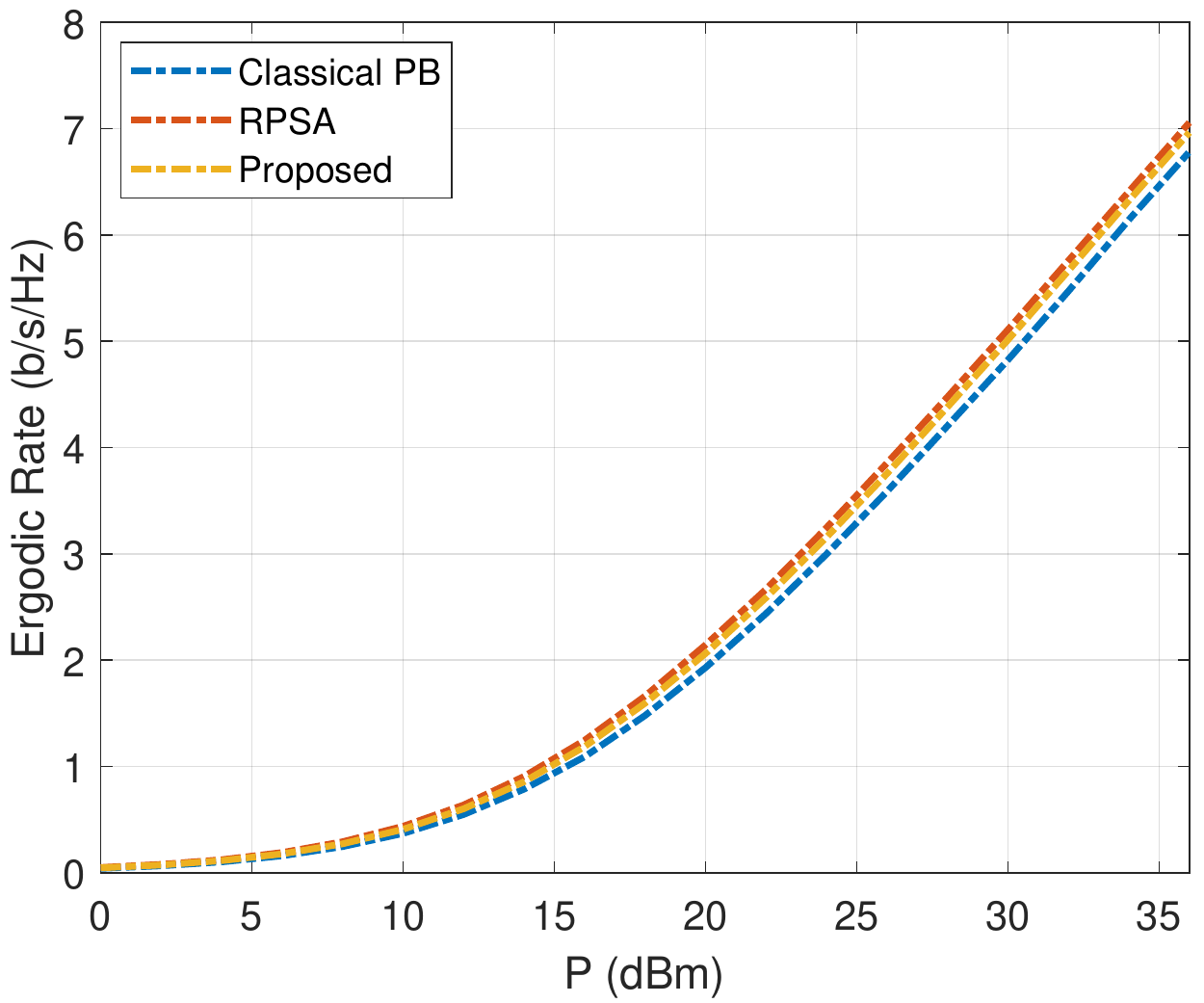}}	\subfloat[]{\includegraphics[width=44mm,height=44mm]{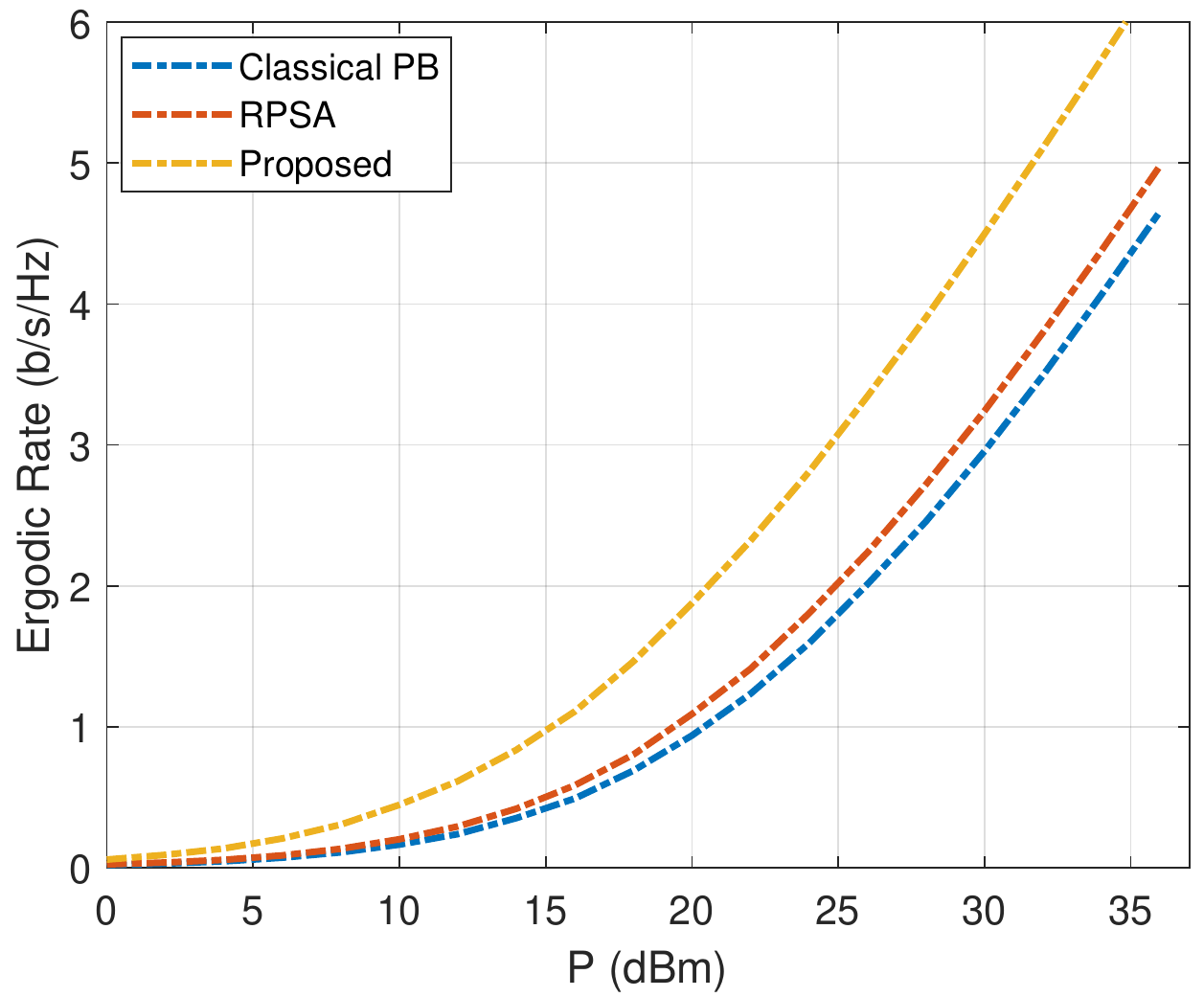}}
	\caption{Ergodic rate performance with (a) no spatial correlation and no phase errors, (b) spatial correlation ($d_H=d_V=d=\lambda/8$) and phase errors ($\kappa=0$).}
	\label{fig:Rate_Err}
	\vspace{-0.7cm}
\end{figure}
\hspace{-0.04cm}It is worth noting that, as shown in Fig. \ref{fig:OP_err}(b) and Fig. \ref{fig:Rate_Err}(b), unlike the proposed scheme, both benchmark schemes suffer under high spatial correlation and phase errors. In particular, the observed behavior of the classical PB under phase errors is reported in \cite{phs-err2-let}, where it is shown that the classical phase shift design (removing the channel phases) may lead to the blind PB scenario, especially when the level of phase errors is high ($\kappa=0$). Also, it is reported in \cite{spatial_ris} that classical phase shift design under spatial correlation leads to a worse performance compared to other phase shift techniques. 

Finally, we verify Proposition 2 in Figs. \ref{fig:ROP}(a) and (b), where $P_a=0.5$ is the asymptotic value of the activation probability obtained from Proposition 1, while the other values, $P_a=0.542$ and $0.5058$, are obtained through simulations for $N=100$ and $5000$, respectively. It can be seen that the upper bound curve $\text{ROP}_\text{U}$ gets closer to the simulation curve when using the value of $P_a$ obtained through simulations. This is due to the convergence of $P_a$ and the approximation of $\text{ROP}_\text{U}$ in Proposition 1 and 2, respectively, requires that $N\rightarrow\infty$. Furthermore, it can be noted that the probability that $N_a$ is within $\pm2\%$ of its mean value increases with $N$, where it is $0.121$ and $0.9102$ for $N=100$ and $5000$, respectively.
\section{Conclusion}\label{sec:Concl}
This study has proposed a low-complexity and novel phase shift-free PB scheme to avoid the hardware implementation and performance degradation issues associated with the classical phase shift-based PB. The proposed PB scheme requires a complexity level linear in $N$ and achieves the same scaling law for the SNR (quadratic growth with $N$) as in the classical phase shift-based PB. Computer simulations show that the proposed PB scheme is superior to the considered benchmark schemes (including the classical PB scheme) in terms of the outage probability and ergodic rate performance. Furthermore, the proposed PB scheme exhibits far less sensitivity to practical system settings, such as the spatial correlation between RIS elements and the phase errors. To conclude, extending the proposed PB scheme to the multiple-input multiple-output systems is an appealing future research direction.
\begin{figure}[t]
	\subfloat[]{\includegraphics[width=44mm,height=45mm]{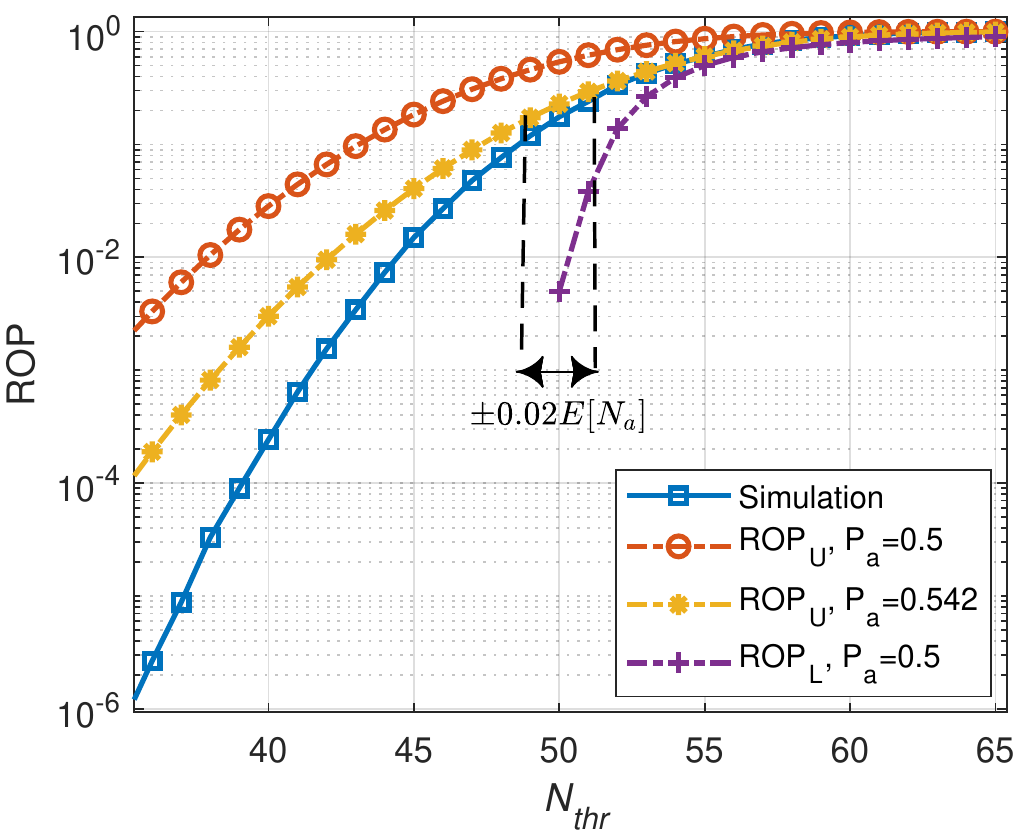}}
	\subfloat[]{\includegraphics[width=44mm,height=45mm]{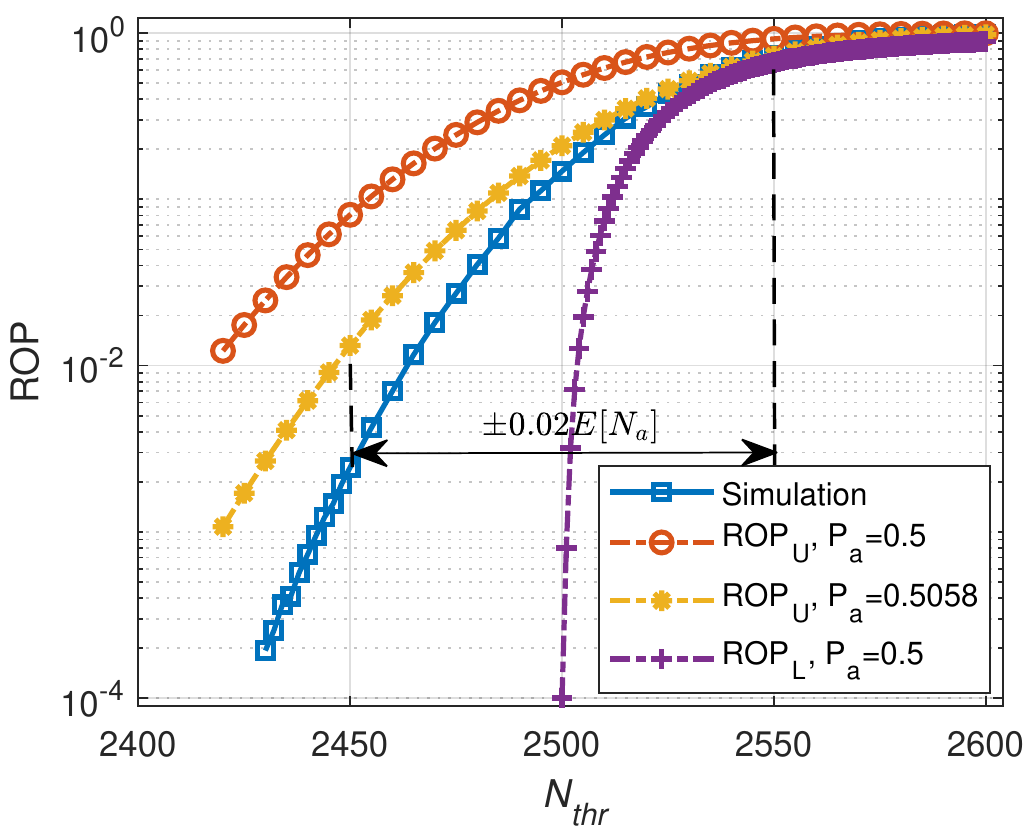}}
	\caption{ROP performance with (a) $N=100$, (b) $N=5000$. }
	\label{fig:ROP}
	\vspace{-0.5cm}
\end{figure}
\appendices
\vspace{-0.2cm}
\section{Proof of Lemma 3}
First, it can be noted that the RVs $\theta^*, Z_1, ..., Z_N$ cannot be mutually independent, where, from \eqref{eq:theta}, $\theta^*$ is a function of $Z_1, ..., Z_N$. Second,
note that $h_1, ..., h_N$ and $g_1, ..., g_N$, by our definition of the S-RIS and RIS-D channels, are mutually independent, which means that the products $h_1g_1, ..., h_Ng_N$ are also mutually independent. Consequently, we conclude that the angles $Z_1, ..., Z_N$ associated with the products are also mutually independent, thus, pairwise independent. In light of these, we need to show that $\theta^*$ and $Z_n$ are asymptotically independent (for any given $n$) to obtain the desired result; $\theta^*, Z_1, ..., Z_N$ are pairwise independent.

Note that the summation in \eqref{eq:theta} can be rewritten as a weighted sum of two unit vectors, 
\begin{align}
	\theta^*&=\arg\left(\sum_{n=1}^{N}w_ne^{jZ_n}\right)\nonumber\\
	&=\arg(w_{N-1}e^{j\bar{Z}_{N-1}}+w_{\tilde{n}}e^{jZ_{\tilde{n}}}),\;\text{for}\;\tilde{n}\neq n,
\end{align}
where $e^{j\bar{Z}_{N-1}}$ and $e^{jZ_{\tilde{n}}}$ are 2D unit vectors, 
$w_{N-1}$ and $w_{\tilde{n}}$ are the wights associated with them, respectively. Here, $w_{N-1}=|\sum_{n=1}^{N-1}h_ng_n|, w_{\tilde{n}}=|h_{\tilde{n}}g_{\tilde{n}}|$, $\bar{Z}_{N-1}=\arg(\sum_{n=1}^{N-1}h_ng_n)$, and $Z_{\tilde{n}}=\arg(h_{\tilde{n}}g_{\tilde{n}})$. Note that the pairs ($w_{N-1}, \bar{Z}_{N-1})$ and $(w_{\tilde{n}}, Z_{\tilde{n}})$  are independent RVs, and $\tilde{n}$ denotes the index of any particular RIS element. In what follows, we show the asymptotic independence of $\theta^*$ and $Z_{\tilde{n}}$ by showing that $w_{\tilde{n}}<<w_{N-1}$ as $N\rightarrow \infty$, which means that, asymptotically, the direction $Z_{\tilde{n}}$ has a trivial effect on the dominant direction $\theta^*$. This can be shown by evaluating the following limit,
\begin{align}
\mathcal{L}&=\lim_{N\rightarrow \infty}P\left(w_{N-1}>tw_{\tilde{n}}\right)\nonumber\\
&=\lim_{N\rightarrow \infty}P\left(w_{N-1}-tw_{\tilde{n}}>0\right)\nonumber\\
&=\lim_{N\rightarrow \infty}1-F_{\bar{w}}(0),
\end{align}
where $t>1$, $\bar{w}=w_{N-1}-tw_{\tilde{n}}$, and $F_{\bar{w}}(0)$ is the CDF of $\bar{w}$ evaluated at zero. Note that, as $N\rightarrow \infty$, the sum $\sum_{n=1}^{N-1}h_ng_n$ converges in distribution to complex Gaussian and thus, $w_{N-1}$ is a Rayleigh RV. However, it is challenging to derive the PDF associated with the distribution of $\bar{w}$, where it corresponds to the difference of a Rayleigh RV and the products of two other Rayleigh RVs. Therefore, we use the Distribution Fitting Tool in MATLAB to fit the distribution of $\bar{w}$ and then obtain its CDF at zero.
\begin{figure}[t]
	\begin{center}
		\includegraphics[width=50mm]{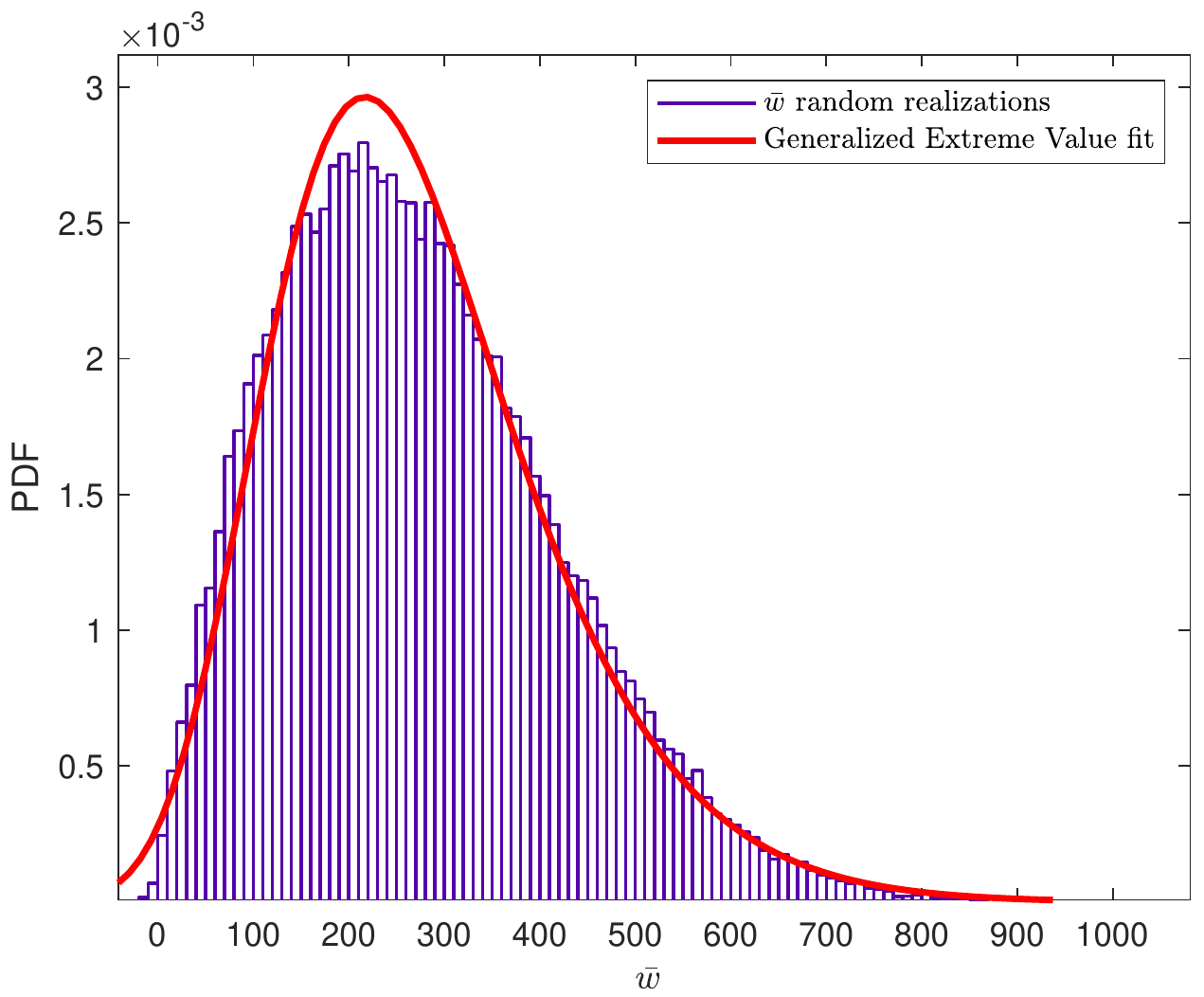}
		\caption{Fitting the distribution of $\bar{w}$ to the generalized extreme value distribution.}		\label{fig:log-normal-fit2}
		\vspace{-0.7cm}
	\end{center}
\end{figure}
As shown in Fig. \ref{fig:log-normal-fit2}, the generalized extreme value (GEV) distribution closely matches the distribution of $\bar{w}$, with $N=10^5$ and the scale, location, and shape parameters are $k=-0.0717$, $\sigma=124.464 $, and $\mu=207.635$, respectively, and without loss of generality, $t=10$. Finally, by evaluating the CDF of the GEV RV with the obtained fitting parameters, we get $F_{\bar{w}}(0)=0.008$ and $\mathcal{L}=0.992$. This shows that the weight $w_{\tilde{n}}<<w_{N-1}$ as $N\rightarrow \infty$, where the dominant direction $\theta^*$ does not depend on a specific direction $Z_{\tilde{n}}$, but, on the combination of all $N$ directions.
Finally, we obtain that $\theta^*, Z_1, ..., Z_N$ are pairwise independent RVs, which completes the proof.\hspace{7.2cm}\qedsymbol  
\section{Proof of Proposition 1}
Note that, for the given $n^{th}$ element, its activation probability has the minimum value when $\theta^*$ is independent of $Z_n$, which occurs at the asymptotic range when $N\rightarrow\infty$ according to Lemma 3. In contrary to this, the activation probability is higher for smaller $N$ values; for example, it can be readily verified from \eqref{eq:theta}  that the activation probability for $N=1$ is always one. In what follows, we show that the asymptotic value of the activation probability is $0.5$, which corresponds to the minimum value as decried earlier.

Note that the $n^{th}$ RIS element is activated ($\eta_n=1$) according to Algorithm 1 when the phase $Z_n$ associated with its channel coefficient ($h_ng_n$) lies in the valid (activation) region $R$, which is specified by the endpoints $\Pi(c_1)$ and $\Pi(c_2)$, as shown in Fig. \ref{Fig:Angles}. It can be also noted form Figs. \ref{Fig:Angles}(a)-(b) that, depending on the value of $\Pi(\theta^*)$, the endpoints $\Pi(c_1)$ and $\Pi(c_2)$ of $R$ switch their positions. Therefore, in what follows, we calculate the activation probability by considering the two cases shown in Figs. \ref{Fig:Angles}(a)-(b), separately.

Let $P_{\Pi(c_1)\geq\Pi(c_2)}$ and $P_{\Pi(c_1)\leq\Pi(c_2)}$ denote the activation probability for the first and second cases shown in Figs. \ref{Fig:Angles}(a) and (b), respectively, then we have
{\small
\begin{align}
	P_{\Pi(c_1)\geq\Pi(c_2)}&=P(\Pi(Z_n)\in[\Pi(c_2),\Pi(c_1)], \Pi(c_1)\geq\Pi(c_2)),\\
	P_{\Pi(c_1)\leq\Pi(c_2)}&=P(\Pi(Z_n)\in[\Pi(c_1),\Pi(c_2)], \Pi(c_1)\leq\Pi(c_2)),
\end{align}
}%
\vspace{-0.3cm}
consequently, using the total probability law, we obtain
\vspace{0.2cm}
{\small
\begin{align}
P_a=	P_{\Pi(c_1)\geq\Pi(c_2)}+P_{\Pi(c_1)\leq\Pi(c_2)}.\label{eq:P_a}
\end{align}
}%

\indent In order to calculate $P_{\Pi(c_1)\geq\Pi(c_2)}$, we note that, as shown in Fig. \ref{Fig:Angles}(a), the activation region $R=R_1\cup R_2$, where $R_1$ and $R_2$ are the two disjoint regions specified by the pairs of endpoints $[\Pi(c_1),0]$ and $[0, \Pi(c_2)]$, respectively. Hence, we obtain
{\small
\begin{align}
	P_{\Pi(c_1)\geq\Pi(c_2)}&=P(\Pi(Z_n)\geq\Pi(c_1))+P(\Pi(Z_n)\leq\Pi(c_2)).\label{eq:P_c1_g_c2}
\end{align}
}%
\begin{figure}[t]
	\centering
	\subfloat[]{\includegraphics[width=40mm]{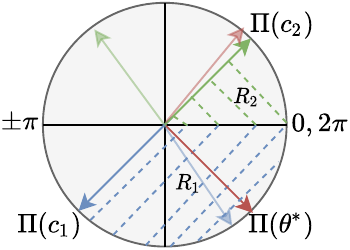}}
	\hspace{0.5cm}	\subfloat[]{\includegraphics[width=40mm]{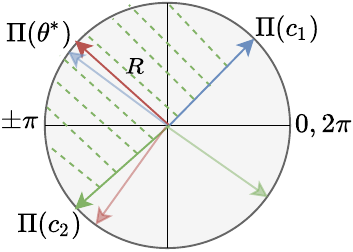}}
	\caption{The valid (activation) region, shown with dash lines, for the case (a) $\Pi(c_1)\geq\Pi(c_2)$ (b) $\Pi(c_1)\leq\Pi(c_2)$.}
	\label{Fig:Angles}
\vspace{-0.5cm}
\end{figure}\indent In order to calculate $P(\Pi(Z_n)\geq\Pi(c_1))$, we note that when $\Pi(c_1)\geq\Pi(c_2)$, we have $\theta^*$ in one of the two intervals $[-\frac{\pi}{2},0]$ or $[0, \frac{\pi}{2}]$. On the other side, from Lemma 1, $Z_n$ lies in one of the two intervals $[-2\pi,0]$ or $[0, 2\pi]$. Also note that, according to \eqref{eq:pi_fun}, $\pm2\pi$ is added whenever $Z_n$ and/or $\theta^*\pm\frac{\pi}{2}$ is out of the natural range of the angle, $[0,2\pi)$. In light of these, by considering the four different combinations of the two intervals associated with $\theta^*$ and $Z_n$, we obtain
{\small
\begin{align}
P(\Pi(Z_n)\geq\Pi(c_1))&=P_{A1}+P_{A2}+P_{A3}+P_{A4}.\label{eq:z_g_c1}
\end{align}
}%
We define the probabilities $P_{Ai}$ as follows for $i\in\{1,..., 4\}$:
{\small
\begin{align}
P_{A1}&=P(Z_n\geq \theta^*-\frac{\pi}{2}+2\pi, \theta^*\in[-\frac{\pi}{2},0], Z_n\in[0,2\pi])\nonumber\\
&=\int_{-\frac{\pi}{2}}^{0}\int_{\theta+\frac{3\pi}{2}}^{2\pi}f_{Z_n,\theta^*}(z,\theta)d_{z}d_{\theta},
\end{align}
}%
by considering Lemma 3, as $N\rightarrow\infty$, we have
{\small
\begin{align}
P_{A1}&=\int_{-\frac{\pi}{2}}^{0}\int_{\theta+\frac{3\pi}{2}}^{2\pi}f_{Z_n}(z)f_{\theta^*}(\theta)d_zd_{\theta}\nonumber\\
&=\int_{-\frac{\pi}{2}}^{0}\int_{\theta+\frac{3\pi}{2}}^{2\pi}\left(\frac{-z+2\pi}{4\pi^2}\right)\left(\frac{1}{2\pi}\right)d_zd_{\theta}=0.0182,
\end{align}
}%
in the same way, we obtain
{\small
\begin{align}
	P_{A2}&=P(Z_n\geq \theta^*-\frac{\pi}{2}+2\pi, \theta^*\in[0,\frac{\pi}{2}], Z_n\in[0,2\pi])\nonumber\\
	&=\int_{0}^{\frac{\pi}{2}}\int_{\theta+\frac{3\pi}{2}}^{2\pi}\left(\frac{-z+2\pi}{4\pi^2}\right)\left(\frac{1}{2\pi}\right)d_zd_{\theta}=0.0026,
\end{align}
}%
{\small
\begin{align}
P_{A3}&=P(Z_n+2\pi\geq \theta^*-\frac{\pi}{2}+2\pi, \theta^*\in[-\frac{\pi}{2},0], Z_n\in[-2\pi,0])\nonumber\\
&=\int_{-\frac{\pi}{2}}^{0}\int_{\theta-\frac{\pi}{2}}^{0}\left(\frac{z+2\pi}{4\pi^2}\right)\left(\frac{1}{2\pi}\right)d_zd_{\theta}=0.0755,
\end{align}
}%
{\small
\begin{align}
	P_{A4}&=P(Z_n+2\pi\geq \theta^*-\frac{\pi}{2}+2\pi, \theta^*\in[0,\frac{\pi}{2}], Z_n\in[-2\pi,0])\nonumber\\
	&=\int_{0}^{\frac{\pi}{2}}\int_{\theta-\frac{\pi}{2}}^{0}\left(\frac{z+2\pi}{4\pi^2}\right)\left(\frac{1}{2\pi}\right)d_zd_{\theta}=0.0286,
\end{align}
}%
thus, from \eqref{eq:z_g_c1}, we obtain $P(\Pi(Z_n)\geq\Pi(c_1))=0.125$.

For $P(\Pi(Z_n)\leq\Pi(c_2))$, by following the same steps above, we have
{\small
\begin{align}
		P(\Pi(Z_n)\leq\Pi(c_2))&=\int_{-\frac{\pi}{2}}^{0}\hspace{-0.08cm}\int_{0}^{\theta+\frac{\pi}{2}}\hspace{-0.1cm}\tilde{f}(z)d_zd_{\theta}+\hspace{-0.1cm}\int_{0}^{\frac{\pi}{2}}\hspace{-0.2cm}\hspace{-0.05cm}\int_{0}^{\theta+\frac{\pi}{2}}\hspace{-0.1cm}\tilde{f}(z)d_zd_{\theta}\nonumber\\	&+\hspace{-0.1cm}\int_{-\frac{\pi}{2}}^{0}\hspace{-0.08cm}\int_{-2\pi}^{\theta-\frac{3\pi}{2}}\hspace{-0.15cm}f^{'}\hspace{-0.1cm}(z)d_zd_{\theta}\hspace{-0.03cm}+\hspace{-0.03cm}\hspace{-0.1cm}\int_{0}^{\frac{\pi}{2}}\hspace{-0.2cm}\int_{-2\pi}^{\theta-\frac{3\pi}{2}}\hspace{-0.15cm}f^{'}\hspace{-0.1cm}(z)d_zd_{\theta}\nonumber\\
		&=0.125,
\end{align}
}%
where, $f^{'}(z)=(\frac{z+2\pi}{4\pi^2})(\frac{1}{2\pi})=\frac{z+2\pi}{8\pi^3}$ and $\tilde{f}(z)=(\frac{-z+2\pi}{4\pi^2})(\frac{1}{2\pi})=\frac{-z+2\pi}{8\pi^3}$.
Finally, from \eqref{eq:P_c1_g_c2}, we have $P_{\Pi(c_1)\geq\Pi(c_2)}=0.25$.

\vspace{0.15cm}
In order to calculate $P_{\Pi(c_1)\leq\Pi(c_2)}$, we note that there is only one activation (valid) region denoted by $R$, as shown in Fig. \ref{Fig:Angles}(b). Therefore, the  activation probability can be given as
{\small
\begin{align}
	P_{\Pi(c_1)\leq\Pi(c_2)}&=P(\Pi(Z_n)\in[\Pi(c_1),\Pi(c_2)])\nonumber\\
	&=P(\Pi(Z_n)\leq\Pi(c_2))-P(\Pi(Z_n)\leq\Pi(c_1))
	,\label{eq:P_c1_l_c2}
\end{align}
}%
where, following the same procedure to compute \eqref{eq:z_g_c1}, we calculate $P(\Pi(Z_n)\leq\Pi(c_1))$ as follows:
{\small
\begin{align}
P(\Pi(Z_n)\leq\Pi(c_1))&\hspace{-0.05cm}=\hspace{-0.05cm}\hspace{-0.1cm}\int_{\frac{\pi}{2}}^{\pi}\int_{0}^{\theta-\frac{\pi}{2}}\hspace{-0.1cm}\tilde{f}(z)d_zd_{\theta}+\hspace{-0.1cm}\int_{\frac{\pi}{2}}^{\pi}\int_{-2\pi}^{\theta-\frac{5\pi}{2}}\hspace{-0.1cm}f{'}(z)d_zd_{\theta}\nonumber\\
&+\hspace{-0.1cm}\int_{-\pi}^{\frac{-\pi}{2}}\hspace{-0.2cm}\int_{0}^{\theta+\frac{3\pi}{2}}\hspace{-0.2cm}\tilde{f}(z)d_zd_{\theta}\hspace{-0.05cm}+\hspace{-0.1cm}\int_{-\pi}^{\frac{-\pi}{2}}\hspace{-0.2cm}\int_{-2\pi}^{\theta-\frac{\pi}{2}}\hspace{-0.2cm}f{'}(z)d_zd_{\theta}\nonumber\\
&=0.125,
\end{align}
}%
and,
\vspace{-0.5cm}
{\small
\begin{align} P(\Pi(Z_n)\leq\Pi(c_2))&\hspace{-0.05cm}=\hspace{-0.05cm}\hspace{-0.05cm}\int_{\frac{\pi}{2}}^{\pi}\hspace{-0.1cm}\int_{0}^{\theta+\frac{\pi}{2}}\hspace{-0.2cm}\tilde{f}(z)d_zd_{\theta}+\hspace{-0.1cm}\int_{\frac{\pi}{2}}^{\pi}\hspace{-0.1cm}\int_{-2\pi}^{\theta-\frac{3\pi}{2}}\hspace{-0.1cm}f{'}(z)d_zd_{\theta}\nonumber\\
	&+\hspace{-0.1cm}\int_{-\pi}^{\frac{-\pi}{2}}\hspace{-0.1cm}\int_{0}^{\theta+\frac{5\pi}{2}}\hspace{-0.25cm}\tilde{f}(z)d_zd_{\theta}\nonumber+\hspace{-0.1cm}\int_{-\pi}^{\frac{-\pi}{2}}\hspace{-0.2cm}\int_{-2\pi}^{\theta+\frac{\pi}{2}}\hspace{-0.2cm}f{'}(z)d_zd_{\theta}\nonumber\\
	&=0.3750,
\end{align}
}%
thus from \eqref{eq:P_c1_l_c2}, we obtain	$P_{\Pi(c_1)\leq\Pi(c_2)}=0.375-0.125=0.25$.
Finally, from \eqref{eq:P_a}, we obtain $P_a=0.25+0.25=0.5$, which corresponds to the minimum activation probability as described before, therefore, we obtain $P_a\geq0.5$. This completes the proof.
\hspace{5.5cm}\qedsymbol  
\section{Proof of Proposition 5}
Note that, from Corollary 1, $N_a$ increases linearly with $N$ in a probabilistic manner. In what follows, we show that the SNR has a quadratic growth with $N_a$ and thus, with $N$.

Let $A_{n^*}=h_{n^*}g_{n^*}$ denote the $n^*$th term of the sum associated with the SNR in \eqref{eq:SNR}, which corresponds to a 2D vector in the complex plane. Consider the length associated with the sum of two given vectors $|A_{\tilde{n^*}}+A_{n^*}|=\sqrt{|A_{\tilde{n^*}}|^2+|A_{n^*}|^2+2|A_{\tilde{n^*}}||A_{n^*}|\cos(\bar{Z}_{n^*})}=d_n^*$, where $n^*, \tilde{n^*}\in\{1, ..., N_a\}$, $n^*\neq \tilde{n^*}$, denote the indices of a pair of two activated vectors and $\bar{Z}_{n^*}=|\arg(A_{n^*})-\arg(A_{\tilde{n^*}})|$ is the angle between them. Considering the worst-case separation scenario of two non-zero length vectors where $\pi/2<\bar{Z}_{n^*}\leq\pi$, then we obtain $d_n^*\leq\min(|A_{\tilde{n^*}}|,|A_{n^*}|))$ with $d_n\rightarrow 0$ as $\bar{Z}_{n^*}\rightarrow \pi$ and $||A_{\tilde{n^*}}|-|A_{n^*}||\rightarrow 0$. We generalize the previous case on all vectors, that is, let $N_a$ be an even number, and the activated vectors can be grouped into $N_a/2$ pairs of two vectors such that $|A_{n^*}|=|A_{\tilde{n^*}}|=A>0, \bar{Z}_{n^*}=\pi-\delta, \forall n^*, \tilde{n^*}\in\{1, ..., N_a\}$, where $\delta$ is an arbitrarily small positive number. Furthermore, without loss of generality, let each pair of vectors be symmetric around the real axis (zero phase), and the vectors on either side of the axis lie on top of each other. Note that this distribution of the vectors corresponds to the worst-case scenario, where, as $\delta\rightarrow 0$, the amplitude of the sum of all vectors converges to zero. In light of these, we get $|\sum_{n^*=1}^{N_a} A_{n^*}|=|\sum_{n^*=1}^{N_a} \Re\{A_{n^*}\}|>0$, where $\Re\{A_{n^*}\}$ denotes the real part of $A_{n^*}$, furthermore, all the terms of the sum have the same sign. This shows that $|\sum_{n^*=1}^{N_a} A_{n^*}|$ grows in an order of $N_a$, which also implies the linear growth with $N$, and, consequently, the SNR in \eqref{eq:SNR} grows in an order of $N^2$. Furthermore, from Algorithm 1, it can be seen that the search loop cannot exceed $2N$ iterations, thus; a complexity level of $\mathcal{O}(N)$ is obtained. This completes the proof.\hspace{7.8cm}\qedsymbol  
\vspace{-0.4cm}
%
\bibliographystyle{IEEEtran}
\bibliography{IEEEabrv,Bibliography}
\vspace{-2cm}
\begin{IEEEbiography}[{\includegraphics[width=1in,height=1.25in,clip,keepaspectratio]{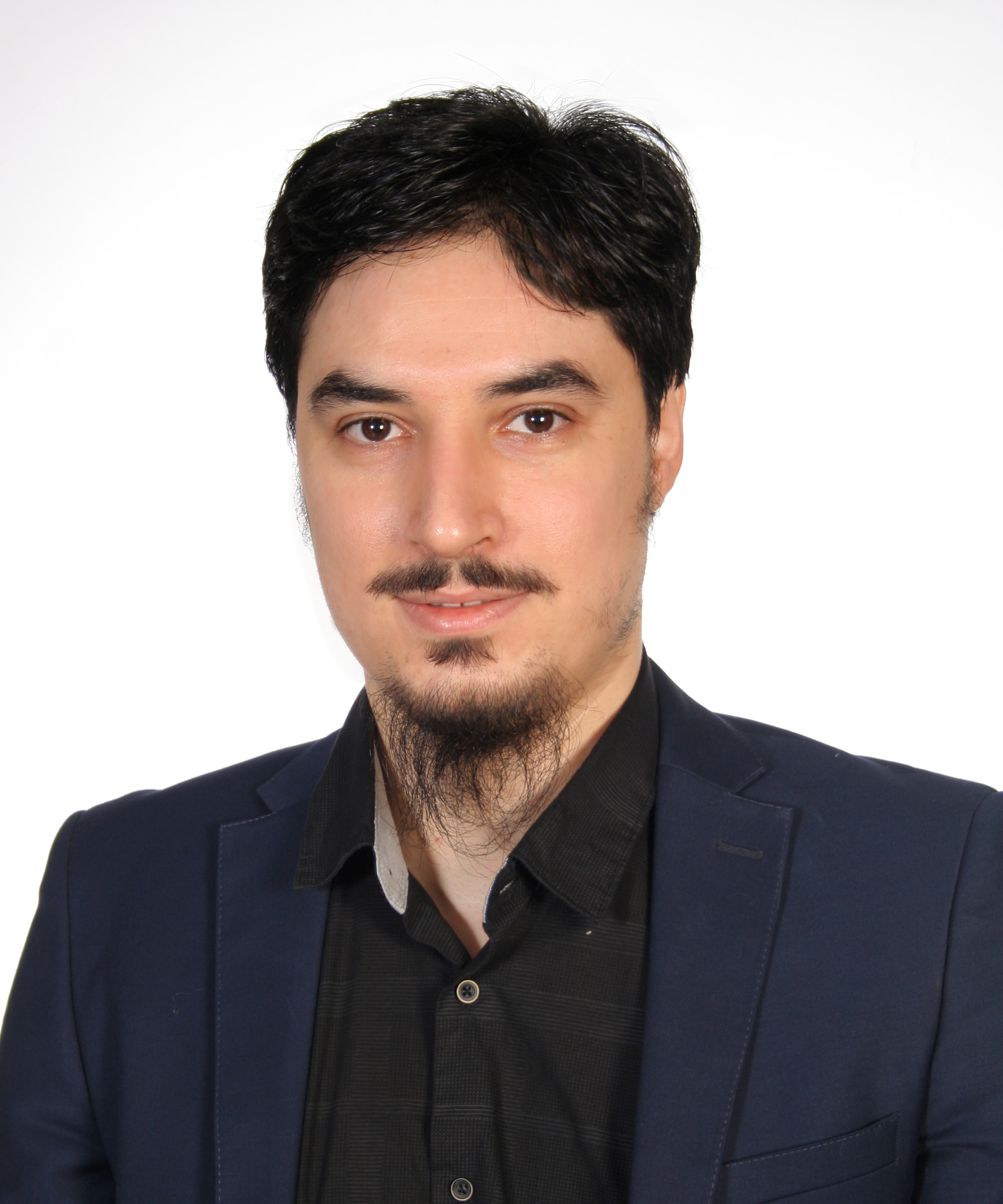}}]{Aymen Khaleel {\normalfont received the B.Sc. degree from the University of Anbar, Al Anbar, Iraq, in 2013, and the M.Sc. degree from Turkish Aeronautical Association University, Ankara, Turkey, in 2017. He is currently pursuing his Ph.D. in Electrical and Electronics Engineering at Ko\c{c} University, Istanbul, Turkey, where he is currently a Project Engineer. His research interests include MIMO systems, index modulation, reconfigurable intelligent surfaces-based systems. He serves as a Reviewer for  \textit{IEEE Transactions on Wireless Communications}, \textit{IEEE Transactions on Vehicular Technology}, \textit{IEEE Communications Magazine}, \textit{IEEE Wireless Communications Letters}, and \textit{IEEE Communications Letters}.}}
\end{IEEEbiography}
\vspace{-1.5cm}
\begin{IEEEbiography}[{\includegraphics[width=1in,height=1.25in,clip,keepaspectratio]{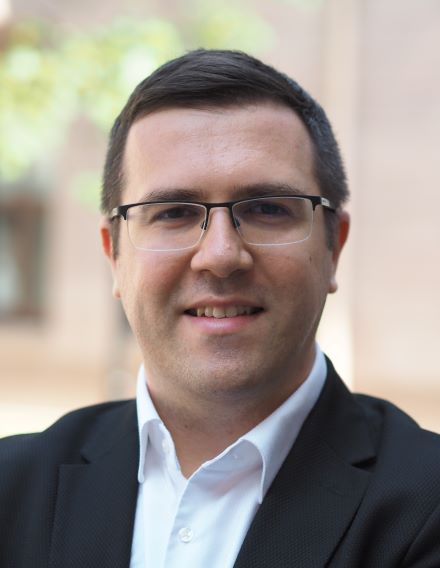}}]{Ertugrul Basar {\normalfont received his Ph.D. degree from Istanbul Technical University in 2013. He is currently an Associate Professor with the Department of Electrical and Electronics Engineering, Ko\c{c} University, Istanbul, Turkey and the director of Communications Research and Innovation Laboratory (CoreLab). His primary research interests include beyond 5G systems, index modulation, intelligent surfaces, waveform design, and signal processing for communications. Dr. Basar currently serves as a Senior Editor of \textit{IEEE Communications Letters} and an Editor of \textit{IEEE Transactions on Communications} and \textit{Frontiers in Communications and Networks}. He is a Young Member of Turkish Academy of Sciences and a Senior Member of IEEE.}}
	
\end{IEEEbiography}
\end{document}